\documentclass{preprint}
\usepackage{helvet}
\usepackage{courier}
\usepackage[T1]{fontenc}
\usepackage[a4paper]{geometry}
\usepackage[active]{srcltx}
\usepackage{color}
\usepackage{verbatim}
\usepackage{mathrsfs}
\usepackage{mathtools}
\usepackage{dsfont}
\usepackage{amsmath}
\usepackage{amsthm}
\usepackage{amssymb}
\usepackage{makeidx}
\makeindex
\usepackage[pdfusetitle,
 bookmarks=true,bookmarksnumbered=true,bookmarksopen=true,bookmarksopenlevel=2,
 breaklinks=false,pdfborder={0 0 1},backref=page,colorlinks=false]
 {hyperref}
\hypersetup{
 unicode=false, linkcolor=black, citecolor=black, urlcolor=blue, filecolor=blue, pdfpagelayout=OneColumn, pdfnewwindow=true, pdfstartview=XYZ, plainpages=false}

\makeatletter
\numberwithin{equation}{section}
\theoremstyle{plain}
\newtheorem{thm}{\protect\theoremname}[section]
\theoremstyle{definition}
\newtheorem{defn}[thm]{\protect\definitionname}
\theoremstyle{plain}
\newtheorem{prop}[thm]{\protect\propositionname}
\theoremstyle{plain}
\newtheorem{lem}[thm]{\protect\lemmaname}
\theoremstyle{plain}
\newtheorem{cor}[thm]{\protect\corollaryname}
\theoremstyle{remark}
\newtheorem{rem}[thm]{\protect\remarkname}
\theoremstyle{definition}
\newtheorem{problem}[thm]{\protect\problemname}

\usepackage{color}
\usepackage{slashed}

 \let\myTOC\tableofcontents
 \renewcommand\tableofcontents{%
   \pdfbookmark[1]{Contents}{}
   \myTOC
   \cleardoublepage
   \pagenumbering{arabic} }

\global\long\def\foreignlanguage#1#2{#2}%
\global\long\def\selectlanguage#1{}%

\allowdisplaybreaks

\DeclareMathOperator\supp{\mathrm{supp}}
\DeclareMathOperator\Tr{\mathrm{Tr}}

\makeatother

\providecommand{\corollaryname}{Corollary}
\providecommand{\definitionname}{Definition}
\providecommand{\lemmaname}{Lemma}
\providecommand{\problemname}{Problem}
\providecommand{\propositionname}{Proposition}
\providecommand{\remarkname}{Remark}
\providecommand{\theoremname}{Theorem}

\begin{document}
\title{The empirical laws for Majorana fields in a curved spacetime}
\author{Hideyasu Yamashita}
\institute{Division of Liberal Arts and Sciences, Aichi-Gakuin University\\
\email{yamasita@dpc.aichi-gakuin.ac.jp}}

\date{\today}

\maketitle
\newcommand{\dcolor}{\definecolor{note_fontcolor}{rgb}{0.1, 0.0, 0.8}}
\definecolor{HYnote_fontcolor}{rgb}{0.2, 0.0, 0.2}
\definecolor{hycolor}{rgb}{0.3, 0.0, 0.3}
\newcommand{\hyc}{\color{hycolor}}
\newenvironment{HYnote}
 {\textcolor{note_fontcolor}\bgroup\ignorespaces}
  {\ignorespacesafterend\egroup}

\newenvironment{trivenv}
  {\bgroup\ignorespaces}
  {\ignorespacesafterend\egroup}

\newcommand{\displabel}[1]{}

\newcommand{\hidable}[3]{#2}
\newcommand{\hidea}[1]{{#1}}
\newcommand{\hideb}[1]{{#1}}
\newcommand{\hidec}[1]{{#1}}
\newcommand{\hidep}[1]{{#1}}
\renewcommand{\hidec}[1]{}
\renewcommand{\hidep}[1]{}

\newcommand{\thlab}[1]{{\tt [#1]}}

\newcommand{\black}{\color{black}}

\global\long\def\N{\mathbb{N}}%
\global\long\def\C{\mathbb{C}}%
\global\long\def\Z{\mathbb{Z}}%

\global\long\def\R{\mathbb{R}}%

\global\long\def\im{\mathrm{i}}%

\global\long\def\di{\partial}%

\global\long\def\d{{\rm d}}%

\global\long\def\ol#1{\overline{#1}}%
\global\long\def\ul#1{\underline{#1}}%
\global\long\def\ob#1{\overbrace{#1}}%

\global\long\def\ov#1{\overline{#1}}%

\global\long\def\then{\Rightarrow}%

\global\long\def\Then{\Longrightarrow}%

\global\long\def\N{\mathbb{N}}%
\global\long\def\C{\mathbb{C}}%
\global\long\def\Z{\mathbb{Z}}%

\global\long\def\R{\mathbb{R}}%

\global\long\def\im{\mathrm{i}}%

\global\long\def\di{\partial}%

\global\long\def\d{{\rm d}}%

\global\long\def\ol#1{\overline{#1}}%
\global\long\def\ul#1{\underline{#1}}%
\global\long\def\ob#1{\overbrace{#1}}%

\global\long\def\ov#1{\overline{#1}}%

\global\long\def\then{\Rightarrow}%

\global\long\def\Then{\Longrightarrow}%

\global\long\def\cA{\mathcal{A}}%
\global\long\def\cB{\mathcal{B}}%

\global\long\def\cC{\mathcal{C}}%

\global\long\def\cD{\mathcal{D}}%
\global\long\def\cE{\mathcal{E}}%

\global\long\def\cF{\mathcal{F}}%

\global\long\def\cG{{\cal G}}%

\global\long\def\cH{\mathcal{H}}%

\global\long\def\cI{\mathcal{I}}%

\global\long\def\cJ{\mathcal{J}}%
\global\long\def\cK{\mathcal{K}}%

\global\long\def\cL{\mathcal{L}}%

\global\long\def\cM{\mathcal{M}}%

\global\long\def\cN{\mathcal{N}}%

\global\long\def\cO{\mathcal{O}}%

\global\long\def\cP{\mathcal{P}}%

\global\long\def\cQ{\mathcal{Q}}%

\global\long\def\cR{\mathcal{R}}%

\global\long\def\cS{\mathcal{S}}%

\global\long\def\cT{\mathcal{T}}%

\global\long\def\cU{\mathcal{U}}%

\global\long\def\cV{\mathcal{V}}%

\global\long\def\cW{\mathcal{W}}%
\global\long\def\cX{\mathcal{X}}%

\global\long\def\cY{\mathcal{Y}}%

\global\long\def\cZ{\mathcal{Z}}%

\global\long\def\scA{\mathscr{A}}%
\global\long\def\scB{\mathscr{B}}%
\global\long\def\scC{\mathscr{C}}%
\global\long\def\scD{\mathscr{D}}%

\global\long\def\scE{\mathscr{E}}%

\global\long\def\scF{\mathscr{F}}%

\global\long\def\scG{\mathscr{G}}%

\global\long\def\scH{\mathscr{H}}%

\global\long\def\scI{\mathscr{I}}%

\global\long\def\scJ{\mathscr{J}}%

\global\long\def\scK{\mathscr{K}}%

\global\long\def\scL{\mathscr{L}}%

\global\long\def\scM{\mathscr{M}}%

\global\long\def\scN{\mathscr{N}}%

\global\long\def\scO{\mathscr{O}}%

\global\long\def\scP{\mathscr{P}}%

\global\long\def\scR{\mathscr{R}}%
\global\long\def\scS{\mathscr{S}}%

\global\long\def\scT{\mathscr{T}}%

\global\long\def\scU{\mathscr{U}}%

\global\long\def\scW{\mathscr{W}}%
\global\long\def\scZ{\mathscr{Z}}%

\global\long\def\bbA{\mathbb{A}}%

\global\long\def\bbB{\mathbb{B}}%

\global\long\def\bbD{\mathbb{D}}%

\global\long\def\bbE{\mathbb{E}}%

\global\long\def\bbF{\mathbb{F}}%

\global\long\def\bbG{\mathbb{G}}%

\global\long\def\bbI{\mathbb{I}}%

\global\long\def\bbJ{\mathbb{J}}%

\global\long\def\bbK{\mathbb{K}}%

\global\long\def\bbL{\mathbb{L}}%

\global\long\def\bbM{\mathbb{M}}%

\global\long\def\bbP{\mathbb{P}}%

\global\long\def\bbQ{\mathbb{Q}}%

\global\long\def\bbT{\mathbb{T}}%

\global\long\def\bbU{\mathbb{U}}%

\global\long\def\bbX{\mathbb{X}}%

\global\long\def\bbY{\mathbb{Y}}%
\global\long\def\bbW{\mathbb{W}}%

\global\long\def\bbOne{1\kern-0.7ex  1}%

\renewcommand{\bbOne}{\mathbbm{1}}

\global\long\def\bB{\mathbf{B}}%

\global\long\def\bG{\mathbf{G}}%

\global\long\def\bH{\mathbf{H}}%
\global\long\def\bS{\boldsymbol{S}}%

\global\long\def\bT{\mathbf{T}}%

\global\long\def\bX{\mathbf{X}}%
\global\long\def\bY{\mathbf{Y}}%
\global\long\def\bW{\mathbf{W}}%

\global\long\def\boT{\boldsymbol{T}}%

\global\long\def\fraka{\mathfrak{a}}%

\global\long\def\frakb{\mathfrak{b}}%

\global\long\def\frakc{\mathfrak{c}}%

\global\long\def\frake{\mathfrak{e}}%

\global\long\def\frakf{\mathfrak{f}}%

\global\long\def\fg{\mathfrak{g}}%

\global\long\def\frakh{\mathfrak{h}}%

\global\long\def\fraki{\mathfrak{i}}%
\global\long\def\frakk{\mathfrak{k}}%

\global\long\def\frakl{\mathfrak{l}}%

\global\long\def\frakm{\mathfrak{m}}%

\global\long\def\frakn{\mathfrak{n}}%

\global\long\def\frako{\mathfrak{o}}%

\global\long\def\frakp{\mathfrak{p}}%

\global\long\def\frakq{\mathfrak{q}}%

\global\long\def\fraks{\mathfrak{s}}%

\global\long\def\fs{\mathfrak{s}}%

\global\long\def\fraku{\mathfrak{u}}%
\global\long\def\frakz{\mathfrak{z}}%

\global\long\def\fA{\mathfrak{A}}%

\global\long\def\fB{\mathfrak{B}}%

\global\long\def\fC{\mathfrak{C}}%

\global\long\def\fD{\mathfrak{D}}%

\global\long\def\fF{\mathfrak{F}}%

\global\long\def\fG{\mathfrak{G}}%

\global\long\def\fK{\mathfrak{K}}%

\global\long\def\fL{\mathfrak{L}}%

\global\long\def\fM{\mathfrak{M}}%

\global\long\def\fP{\mathfrak{P}}%

\global\long\def\fR{\mathfrak{R}}%

\global\long\def\fS{\mathfrak{S}}%
\global\long\def\fT{\mathfrak{T}}%

\global\long\def\fU{\mathfrak{U}}%

\global\long\def\fX{\mathfrak{X}}%

\global\long\def\ssS{\mathsf{S}}%
\global\long\def\ssT{\mathsf{T}}%

\global\long\def\ssW{\mathsf{W}}%

\global\long\def\hM{\hat{M}}%

\global\long\def\rM{\mathrm{M}}%
\global\long\def\prj{\mathfrak{P}}%

{}
\global\long\def\sy#1{{\color{blue}#1}}%

\global\long\def\magenta#1{{\color{magenta}#1}}%

\global\long\def\symb#1{#1}%

\global\long\def\emhrb#1{\text{{\color{red}{\huge {\bf #1}}}}}%

\newcommand{\symbi}[1]{\index{$ #1$}{\color{red}#1}}

{}
\global\long\def\SYM#1#2{#1}%

\renewcommand{\SYM}[2]{\symb{#1}}

\newcommand{\usuji}{\color[rgb]{0.7,0.4,0.4}} \newcommand{\usu}{\color[rgb]{0.5,0.2,0.1}}
\newenvironment{Usuji} {\begin{trivlist}   \item \usuji }  {\end{trivlist}}
\newenvironment{Usu} {\begin{trivlist}   \item \usu }  {\end{trivlist}}

\newcommand{\term}[1]{\textcolor[rgb]{0, 0, 1}{\bf #1}}
\newcommand{\termi}[1]{{\bf #1}}

\newcommand{\slim}{\mathop{\mbox{s-lim}}} %

\newcommand{\wlim}{\mathop{\mbox{w-lim}}}

\newcommand{\limsub}{\mathop{\mbox{\rm lim-sub}}}

\global\long\def\bboxplus{\boxplus}%

\renewcommand{\bboxplus}{\mathop{\raisebox{-0.8ex}{\text{\begin{trivenv}\LARGE{}$\boxplus$\end{trivenv}}}}}

\global\long\def\shuff{\sqcup\kern-0.3ex  \sqcup}%

\renewcommand{\shuff}{\shuffle}

\global\long\def\upha{\upharpoonright}%

\global\long\def\ket#1{|#1\rangle}%

\global\long\def\bra#1{\langle#1|}%

{}
\global\long\def\lll{\vert\kern-0.25ex  \vert\kern-0.25ex  \vert}%
 \renewcommand{\lll}{{\vert\kern-0.25ex  \vert\kern-0.25ex  \vert}}

\global\long\def\biglll{\big\vert\kern-0.25ex  \big\vert\kern-0.25ex  \big\vert\kern-0.25ex  }%

\global\long\def\Biglll{\Big\vert\kern-0.25ex  \Big\vert\kern-0.25ex  \Big\vert}%

\newcommand{\iiia}[1]{{\left\vert\kern-0.25ex\left\vert\kern-0.25ex\left\vert #1
  \right\vert\kern-0.25ex\right\vert\kern-0.25ex\right\vert}}

\global\long\def\iii#1{\iiia{#1}}%

\global\long\def\Upa{\Uparrow}%

\global\long\def\Nor{\Uparrow}%

\newcommand{\vertt}{\kern-0.6ex\vert}
\renewcommand{\Nor}{[\kern-0.16ex ]}

\global\long\def\Prob{\mathbb{P}}%
\global\long\def\Var{\mathrm{Var}}%
\global\long\def\Cov{\mathrm{Cov}}%
\global\long\def\Ex{\mathbb{E}}%
{} 
\global\long\def\Ae{{\rm a.e.}}%

\global\long\def\samples{\bOm}%


\global\long\def\bOne{{\bf 1}}%

\global\long\def\Ten{\bullet}%
{} %

\global\long\def\TT{\intercal}%
 \renewcommand{\TT}{\mathsf{T}}

\global\long\def\trit{\vartriangle\!\! t}%

\global\long\def\Ad{{\rm Ad}}%
\global\long\def\bA{\mathbf{A}}%
\global\long\def\bM{\mathbf{M}}%

\global\long\def\bV{\mathbf{V}}%
\global\long\def\bfS{{\bf S}}%

\global\long\def\bQ{\mathbf{Q}}%

\global\long\def\Borel{{\rm Borel}}%
\global\long\def\bProj{\mathfrak{P}}%
\global\long\def\bE{\mathbf{E}}%

\global\long\def\diam{\text{{\rm diam}}}%
\global\long\def\dom{\mathrm{dom}}%
\global\long\def\End{{\rm End}}%
\global\long\def\ex{{\rm ex}}%

\global\long\def\grad{\mathrm{grad}}%

\global\long\def\Hom{\mathrm{Hom}}%

\global\long\def\Id{{\rm Id}}%

\global\long\def\id{{\rm id}}%
\global\long\def\Inv{{\rm Inv}}%

\global\long\def\Lie{{\rm Lie}}%

\global\long\def\leng{\text{{\rm leng}}}%

\global\long\def\Leb{\text{{\rm Leb}}}%

\global\long\def\meas{\text{{\rm meas}}}%

\global\long\def\GreenOp{\mathsf{E}}%
\global\long\def\Sol{\mathsf{Sol}}%
\global\long\def\GL{{\rm GL}}%

\global\long\def\Pow{\mathsf{P}}%

\global\long\def\p{\mathbf{p}}%

\global\long\def\q{\mathbf{q}}%

\global\long\def\rF{\mathrm{F}}%

\global\long\def\rE{\mathrm{E}}%

\global\long\def\ran{\mathrm{ran}}%

\global\long\def\sfS{\mathsf{S}}%
\global\long\def\sfQ{\mathsf{Q}}%

\global\long\def\Span{{\rm span}}%
\global\long\def\sgn{{\rm sgn}}%

\global\long\def\spec{{\rm spec}}%

\global\long\def\Sp{{\rm Sp}}%
\global\long\def\sfD{\mathsf{D}}%

\global\long\def\sperp{\angle}%
{}

{} %

\global\long\def\WICK#1{{\rm w}\{#1\}}%
 \renewcommand{\WICK}[1]{{:}#1{:}}

\global\long\def\x{\mathbf{x}}%

\global\long\def\y{\mathbf{y}}%

\global\long\def\ad{{\rm ad}}%
\global\long\def\Borel{{\rm Borel}}%
\global\long\def\area{{\bf S}}%
\global\long\def\bbS{\mathbb{S}}%
\global\long\def\bbOne{\mathds{1}}%
\global\long\def\Bdd{\mathscr{B}}%

\global\long\def\Cl{{\rm C}\ell}%
\global\long\def\cconj{\blacklozenge}%
\global\long\def\cpt{{\rm c}}%
\global\long\def\cc{{\bf c}}%
\global\long\def\curve{\mathsf{C}}%

\global\long\def\CAR{{\rm CAR}}%
\global\long\def\ComplexS{J}%

\global\long\def\DiracOpm{{\bf \mathsf{D}}}%
\global\long\def\Dconj{\mathfrak{d}}%
\global\long\def\DSpinors{\mathscr{C}}%

\global\long\def\Vac{\mathsf{E}}%

\global\long\def\GreenOp{\mathsf{S}}%
\global\long\def\hol{{\rm hol}}%

\global\long\def\LBundle{\cL}%

\global\long\def\Manifold{\mathscr{X}}%

\global\long\def\Mat{{\rm Mat}}%

\global\long\def\nablaslash{\slashed{\nabla}}%

\global\long\def\ON{{\rm ON}}%

\global\long\def\OCpl{{\rm OC}}%

\global\long\def\Proj{{\bf Pj}}%
\global\long\def\bP{{\bf P}}%

\global\long\def\qdev{{\rm qd}}%
\global\long\def\Paths{\scE}%

\global\long\def\Pin{{\rm Pin}}%

\global\long\def\rD{{\rm D}}%
\global\long\def\re{{\bf r}}%
\global\long\def\RDSPinors{\scR}%

\global\long\def\Span{{\rm span}}%
\global\long\def\Spin{{\rm Spin}}%
\global\long\def\spin{\mathfrak{spin}}%
\global\long\def\Sol{\mathsf{Sol}_{{\rm sc}}}%

\def\foreignlanguage#1#2{#2}


\let\ruleorig=\rule
\renewcommand{\rule}{\noindent\ruleorig}

\global\long\def\labelenumi{(\arabic{enumi})}%

\begin{abstract}
This article is a sequel to our previous paper (2025), where we considered
the conceptual problem on the empirical laws for the Klein\textendash Gordon
quantum field theory in curved spacetime (QFTCS), and we will consider
the similar problems for the Majorana field on curved spacetime here.
A ``law'' in theoretical physics is said to be \emph{observable
}or\emph{ empirical} only if it can be verified/falsified by some
experimental procedure. The notion of empiricality/observability becomes
far more unclear in QFTCS, than in QFT in Minkowski (flat) spacetime
(QFTM), mainly because QFTCS lacks the notion of vacuum. This could
potentially undermine the status of QFTCS as a \emph{physical} (not
only mathematical) theory. We consider this problem for the Majorana
field in curved spacetime, and examine some examples of the empirical
laws.
\end{abstract}

\section{Introduction}

\label{sec:Introduction}

Our starting point is the following elementary but fundamental problem:
``how can physical/empirical laws be described by a QFT in curved
spacetime (QFTCS)?''. (For rigorous formulations of QFTCS, see \cite{Wal1994,BFV2003,BF2009,HW2015,KM2015,FR2016}
and references therein.) We stressed the importance of this problem
in Introduction of \cite{Yam2025}; Although similar problems can
arise also in QFT in Minkowski spacetime (QFTM), and even in nonrelativistic
quantum mechanics, it seems that QFTCS has far more serious conceptual
problems on empirical laws, mainly because neither S-matrices nor
Wightman functions are defined in QFT in a generic curved spacetime.
For example, it seems to remain an unsolved problem whether the Unruh
effect is an empirical law, that is, whether we can verify/falsify
the Unruh effect experimentally.%
{} See \cite{CHM2007,Ear2011} and references therein. Currently, it
seems that QFTCS has no novel physical results which are proven to
be experimentally verifiable/falsifiable without doubt,\footnote{Although the verifiability/falsifiability of the Hawking effect seems
somewhat better than that of the Unruh effect, we will not be able
to say that the verifiability/falsifiability is without doubt.} although QFTCS is a well-established \emph{mathematical} theory.%
{} Thus there is a significance in searching for the nontrivial examples
of verifiable/falsifiable laws in QFTCS.

This article is a sequel to Yamashita \cite{Yam2025}, where we considered
the above problem for the Klein\textendash Gordon (KG) field in curved
spacetime, and we will consider the similar problems for the Majorana
field on curved spacetime here.

Now we explain the meaning of the term ``empirical law'' used in
this paper. Exactly speaking, a ``law'' in theoretical physics is
said to be \emph{observable }or\emph{ empirical} only if it can be
verified/falsified by some experimental procedure. A typical example
of a ``law'' which is \emph{not} observable/empirical is the equation
$\di^{\mu}A_{\mu}=0$, where $A$ is the vector potential of an electromagnetic
field; This is nothing other than the Lorenz gauge condition for $A$.

The best way to show that a law is empirical is to construct an experimental
procedure to verify/falsify it, concretely. However, actually there
are many so-called ``observable quantities'' in quantum physics
such that it is not easy for us to imagine the measurement processes
of the quantity. For example, let $X$ and $P$ be the position and
the momentum of a (non-relativistic) quantum particle, respectively.
Both $X$ and $P$ are familiar to us, and in fact we can construct
the concrete measure procedure for them, relatively easily. On the
other hand, the ``observable quantity'' $X+rP$ ($r\in\R\setminus\{0\}$
is a constant) is less familiar, and the quantity $XP+PX$ is still
less familiar to us, although we believe that the measurements of
them are possible in principle. Generally speaking, one may hope that
a physical/empirical laws can be written in terms of rather ``familiar''
quantities, so that we can easily find the procedure to verify/falsify
the law.

Hence we would like to choose a set $\scO$ of %
``relatively familiar'' observable quantities, and consider the
elements of $\scO$ only, so that we can describe the physical/empirical
laws in terms of them, hypothesizing the existence of the measurement
procedure of them. Of course, it is preferable that we can construct
the concrete measurement procedure of every element of $\scO$, but
we do not require that here.

Perhaps, in QFT, the observables which are the most ``familiar''
to us are the ones so-called the ``intensity/strength of field'',
precisely expressed e.g., by the stress-energy tensor and the angular
momentum tensor, as well as the 4-current density of charge for the
Dirac field. However their rigorous definitions are rather difficult
in curved spacetime, even for free fields; Actually they need the
mathematical tools from \termi{microlocal analysis} (see \cite{DHP2009,HW2015}
and references therein). In this paper, we shall confine ourselves
to the ``mathematically basic/elementary observables'' which can
be defined without microlocal analysis, at some cost of the ``familiarity''.

To make matters worse, no realistic particles have been proven to
be Majorana particles, experimentally, although there is a possibility
that neutrinos turn out to be Majorana particles. Even if neutrinos
are Majorana particles, it is extremely difficult to \emph{manipulate}
neutrinos; Only a small number of high-energy neutrinos can be detected
destructively. Nevertheless, we will examine the Majorana field rather
than the Dirac field in this paper, since the former is conceptually
simpler in the following sense. Let $\cA_{{\rm Majorana}}$ be the
$C^{*}$-algebra generated by the Majorana field operators. Then the
even subalgebra $\cA_{{\rm Majorana}}^{{\rm even}}\subset\cA_{{\rm Majorana}}$
can be viewed as the algebra of observables. On the other hand, let
$\cA_{{\rm Dirac}}$ be the $C^{*}$-algebra generated by the field
operators of the Dirac field. In this case, there are many non-observables
in $\cA_{{\rm Dirac}}^{{\rm even}}$, because of the $U(1)$ gauge
symmetry (equivalently, of the charge superselection rule) of the
Dirac field.

\section{CAR algebra}

In this paper, we use the symbol $\langle\cdot|\cdot\rangle$ to denote
a sesquilinear form, mainly a (possibly indefinite) Hermitian inner
product on some complex vector space, which is linear in the second
argument, and $(\cdot|\cdot)$ to denote a \emph{real} symmetric bilinear
form on some (real or complex) vector space. %
{} On the other hand, there is no firm convention regarding the usage
of the symbol $\langle\cdot,\cdot\rangle$, but sometimes it denotes
a complex bilinear form. The symbol $(\cdot,\cdot)$ is used to denote
an ordered pair.

We will give two almost equivalent definitions of CAR algebras.
\begin{defn}
\label{def:CAR-real}(CAR algebra: real version) Let $V$ be a real
vector space, and $(\cdot|\cdot)$ a (real) pre-inner product (i.e.,
a positive-semidefinite symmetric bilinear form) on $V$.  %
{} The \termi{CAR algebra} $\SYM{\CAR(V)}{CAR()}=\CAR(V,(\cdot|\cdot))$
over $V$ is defined to be the $C^{*}$-algebra generated by $\bbOne$
and $\phi(v),v\in V$, satisfying
\begin{enumerate}
\item $v\mapsto\phi(v)$ is $\R$-linear,
\item $\phi(v)^{*}=\phi(v),$ for all $v\in V$,
\item $\{\phi(v),\phi(u)\}=2(v|u)\bbOne,$ $v,u\in V$, where $\{X,Y\}:=XY+YX$.
\item If $(v|v)=0$, then $\phi(v)=0$.
\end{enumerate}
for all $u,v\in V$.
\end{defn}

\begin{defn}
\label{def:CAR-complex}(CAR algebra: complex version) Let $\frakh$
be a complex vector space, and $\langle\cdot|\cdot\rangle$ be a complex
pre-inner product (i.e., a positive-semidefinite Hermitian form) on
$\frakh$. %
{} The \termi{CAR algebra} $\SYM{\CAR(\frakh)}{CAR(h)}=\CAR(\frakh,\langle\cdot|\cdot\rangle)$
over $\frakh$ is a $C^{*}$-algebra $\mathfrak{A}$ generated by
the identity $\bbOne$ and elements $a(v),v\in\frakh,$ satisfying
\begin{enumerate}
\item $v\mapsto a(v)$ is antilinear,
\item $\{a(v),a(u)\}=0$,
\item $\{a(v),a(u)^{*}\}=\langle v|u\rangle\bbOne$,
\item If $\langle v|v\rangle=0$, then $a(v)=0$.
\end{enumerate}
for all $v,u\in\frakh.$
\end{defn}

Note that in these definitions, we does not assume that $(\cdot|\cdot)$
and $\langle\cdot|\cdot\rangle$ are inner products, i.e., $(u|u)=0\then u=0$
for $u\in V$ and $\langle v|v\rangle=0\then v=0$ for $v\in\frakh$.
However, the above $\CAR(V)$ is almost the same as $\CAR(V/N)$ where
$N:=\{u\in V:\,(u|u)=0\}$, and similarly for $\CAR(\frakh)$. %
These definitions are well-defined in the sense that $\CAR(V)$ (resp.~$\CAR(\frakh)$)
is shown to be unique up to $*$-isomorphism for any given $V$ (resp.~$\frakh$)
\cite{BR97}.

Consider the CAR algebra $\CAR(\frakh,\langle\cdot|\cdot\rangle)$
over a complex pre-inner product space $(\frakh,\langle\cdot|\cdot\rangle)$.
Then $(\frakh,(\cdot|\cdot))$ becomes a real pre-inner product space,
where $(u|v):=\Re\langle u|v\rangle$, $u,v\in\frakh$. In this case,
we find that $\CAR(\frakh,\langle\cdot|\cdot\rangle)$ is isomorphic
to $\CAR(\frakh,(\cdot|\cdot))$. In fact, the generators $\phi(\cdot)$
of $\CAR(\frakh,(\cdot|\cdot))$ is expressed by the generators $a(\cdot)$
of $\CAR(\frakh,\langle\cdot|\cdot\rangle)$ by
\[
\phi(v):=a(v)+a^{*}(v).
\]
Conversely, consider the CAR algebra $\CAR(V,(\cdot|\cdot))$ over
a real inner product space $(V(\cdot|\cdot))$. %
{} Assume that $V$ has an orthogonal complex structure $\ComplexS$,
i.e., $\ComplexS$ is an operator on $V$ such that $\ComplexS^{2}=-1$
and $(\ComplexS v|\ComplexS u)=(v|u)$. (For example, such $\ComplexS$
exists if $V$ is a real Hilbert space of even (or infinite) dimension.)
Then $V$ can be viewed as a complex linear space $\tilde{V}$ with
a complex inner product $\langle\cdot|\cdot\rangle_{\ComplexS}$.
Then $\CAR(\tilde{V},\langle\cdot|\cdot\rangle_{\ComplexS})$ is isomorphic
to $\CAR(V,(\cdot|\cdot))$ by
\[
a_{\ComplexS}(v):=\frac{1}{2}\left(\phi(v)+\im\phi(\ComplexS v)\right).
\]

In QFT in Minkowski spacetime (QFTM), usually a fermion field theory
begins with the fermion Fock space $\cF_{{\rm f}}(\frakh)$, where
$\frakh$ is a complex Hilbert space. Then we have a representation
of $\CAR(\frakh)$ on $\cF_{{\rm f}}(\frakh)$ such that $a(v)$ and
$a^{*}(v)$ are represented as the annihilation and creation operators
on $\cF_{{\rm f}}(\frakh)$, respectively. Thus the complex version
of the definition of the CAR algebra seems more natural in QFTM. However,
in QFTCS, probably it is not a right way to begin with a Fock space,
which has a (unique) vacuum, and hence the information of the complex
structure $\ComplexS$ is suspected of being physically/empirically
redundant, like a gauge fixing. Therefore, we prefer the real version
of the CAR algebra in this paper.

\section{Definition of the Majorana and Dirac fields in curved spacetime}

In the Minkowski spacetime $\R^{1,3}$, a Majorana/Dirac field%
{} is usually defined on a Fock space, which is viewed as a Hilbert
space furnished with some extra information, e.g., a unique fixed
vacuum, the creation/annihilation operators, the number operator,
etc. On the other hand, in a curved space time $M$, a Majorana/Dirac
field is usually defined as a CAR algebra $\CAR(V,(\cdot|\cdot))$,
without referring to any Hilbert space and any fixed vacuum state
\cite{Dim1982,DH2006,San2008,DHP2009,BD2015}. This section outlines
this definition of the Majorana/Dirac fields in curved spacetime.
The primary aim of this section is to see that a considerable part
of the problems on the empirical laws of the Majorana/Dirac fields
in curved spacetime can be reduced to those problems on general CAR
algebras. %
As well be explained later, the information on the structure of the
spacetime $M$ is encoded in the pre-inner product space $(V,(\cdot|\cdot))$.
In contrast, in the Wightman view of QFTM \cite{SW64}, the most part
of those problems on QFTM are considered to be reduced to the problems
on vacuum expectation values. This Wightman view cannot be maintained
in QFTCS.

The description of this section mainly follows Sanders \cite{San2008},
but our symbols and notations are considerably different from those
of \cite{San2008}. We would like to make a slight reformulation (or
rather a ``purification'' (see below)) of the conventional formalism
for the Dirac and the Majorana fields. For example, the space of Dirac
spinor $\C^{4}$ is usually implicitly supposed to be given two inner
products, namely, the (indefinite) Dirac inner product and the usual
(positive-definite) inner product on $\C^{4}$ \cite{San2008,DHP2009,BD2015}.
However, the information of the latter inner product is considered
to be physically/empirically redundant. Therefore, from our perspective
in this paper, it is reasonable to avoid referring to the latter inner
product, as well as to $X^{\TT}$, $\psi^{\TT}$ (transpose) or $X^{*}$,
$\psi^{*}$ (Hermitian adjoint) of $X\in\Mat(4,\C)$, $\psi\in\C^{4}$
(as a column vector), since their definitions depend on the positive-definite
inner product on $\C^{4}$.

Furthermore, recall that when we consider the empirical meaning of
$\CAR(V)$, also the empirical meaning of the base space (the space
of test functions) $V$ should be clarified, as well as possible.
Thus the space of test functions should not be chosen only for the
mathematical convenience. For example, while Sanders \cite{San2008}
and Dappiaggi\,\&\,Hack\,\&\,Pinamonti \cite{DHP2009} use the
space of the compactly supported sections of the double Dirac spinor
bundle $DM\oplus D^{*}M$ as the test function space, this space seems
to be unnecessarily large, where any element of it has more non-physical
redundant information.\footnote{In any way, it is not easy to empirically interpret the sections of
the Dirac (complex-valued) spinor bundle, due to the $U(1)$ gauge
symmetry. On the other hand, the Majorana (real-valued) spinor sections
are easier to do.} A similar argument applies to D'Antoni\,\&\,Hollands \cite{DH2006}
on the Majorana field.%
{}

From these reasons, we would like to present a somewhat ``purified''
form of the definitions of the Majorana/Dirac field, that is, to reduce
the empirical redundancy from them.

{} %

Note that the CAR algebra $\cA=\CAR(V)$ lacks some important observables,
e.g., the stress-energy tensor, the angular momentum tensor and the
4-current density of charge for the Dirac field. To include such observables,
we need an extension $\cA_{{\rm ext}}$ of the algebra $\cA$ \cite{San2008,DHP2009},
which we do not examine in this paper. It is not immediately clear
whether any empirical law which is stated in terms of $\cA_{{\rm ext}}$
can be restated in terms of $\cA$, in principle. We would like to
examine this problem elsewhere.

\subsection{Dirac matrices}

Define the matrix $I^{(r,s)}\in\Mat(r+s,\C)$ by
\[
\SYM{I^{(r,s)}}{I(r,s)}:={\rm diag}(\overbrace{1,...,1}^{r},\overbrace{-1,...,-1}^{s}),
\]
and let $\eta:=I^{(1,3)}$.%
{} The \termi{Clifford algebra} $\SYM{\Cl_{r,s}}{Clrs}$ is the $\R$-algebra
freely generated by a unit element $\bbOne$ and $e_{1},...,e_{r+s}$
subject to the relations%
\[
e_{i}e_{j}+e_{j}e_{i}=I_{ij}^{(r,s)},\qquad i,j=1,...,r+s
\]

In fact, any choice of $\gamma_{0},\dots,\gamma_{3}\in\Mat(4,\C)$
satisfying $\gamma_{\mu}\gamma_{\nu}+\gamma_{\nu}\gamma_{\mu}=2\eta_{\mu\nu}1_{4}$
for all $\mu,\nu=0,1,2,3$ induces an irreducible representation $\SYM{\pi}{pi}:\Cl_{1,3}\to\Mat(4,\C)$
defined by $\pi(g_{\mu})=\gamma_{\mu}$.

\begin{prop}
[Pauli's fundamental theorem]\label{prop:gamma-uniqueness}Let $\gamma_{0},\dots,\gamma_{3}\in\Mat(4,\C)$
and $\gamma_{0}',\dots,\gamma_{3}'\in\Mat(4,\C)$ be two sets of the
above gamma matrices. Then there exists $L\in GL(4,\C)$ such that
$\gamma_{\mu}'=L\gamma_{\mu}L^{-1}$, $\mu=1,...,4$. Such $L$ is
unique up to scalar factor.
\end{prop}

The \termi{Pin group} $\Pin_{r,s}$ and the \termi{Spin group} $\Spin_{r,s}$
are defined by
\[
\SYM{\Pin_{r,s}}{Pinrs}:=\left\{ S\in\Cl_{r,s}|\ S=u_{1}\cdots u_{k},\ k\in\N,\ u_{i}\in\mathbb{R}^{r,s},\ u_{i}^{2}=\pm\bbOne,\ i=1,...,k\right\} ,
\]
\[
\SYM{\Spin_{r,s}}{Spinrs}:=\Pin_{r,s}\cap\Cl_{r,s}^{{\rm even}}.
\]

Recall the definitions of the Lorentz group, the proper Lorentz group
and the proper orthochronous Lorentz group: They are respectively
given by $\SYM{\mathcal{L}}L:=O_{1,3}$, $\SYM{\mathcal{L}_{+}}{L+}:=SO_{1,3}$
and $\SYM{\mathcal{L}_{+}^{\uparrow}}{L+|}:=\SYM{SO_{1,3}^{0}}{SO1,30}$,
which is the connected component of $\mathcal{L}_{+}$ containing
the identity.

Let $\SYM{\Spin_{1,3}^{0}}{Spin1,30}$ denote the connected component
of $\Spin_{1,3}$ which contains the identity.

Through the representation $\pi$, each $g\in\Pin_{1,3}$ naturally
acts on $\C^{4}$.
\begin{prop}
There exists a (indefinite) Hermitian inner product (i.e., a nondegenerate
Hermitian form) $\langle\cdot|\cdot\rangle_{\rD}$ on $\C^{4}$ such
that%
\begin{enumerate}
\item $\langle\cdot|\cdot\rangle_{\rD}$ is $\Spin_{1,3}^{0}$-invariant;
i.e., $\langle g\psi_{1}|g\psi_{2}\rangle_{\rD}=\langle\psi_{1}|\psi_{2}\rangle_{\rD}$
for all $g\in\Spin_{1,3}^{0}$ and $\psi_{1},\psi_{2}\in\C^{4}$.
\item $\gamma_{a}$ is $\langle\cdot|\cdot\rangle_{\rD}$-symmetric, i.e.,
$\langle\gamma_{a}\psi_{1}|\psi_{2}\rangle_{\rD}=\langle\psi_{1}|\gamma_{a}\psi_{2}\rangle_{\rD}$
for $\psi_{1},\psi_{2}\in\C^{4}$, $a=0,...,3$.
\item For all future pointing time-like vector $n$, $n^{a}\gamma_{a}$
is $\langle\cdot|\cdot\rangle_{\rD}$-positive, i.e., $\langle\psi|n^{a}\gamma_{a}\psi\rangle_{\rD}>0$
for all $\psi\in\C^{4}\setminus\{0\}$.
\end{enumerate}
Such Hermitian inner product $\langle\cdot|\cdot\rangle_{\rD}$ is
unique up to a positive factor
\end{prop}

See e.g., Theorem 4.1.6 and Lemma 4.1.13 in \cite{San2008} for the
proof. Note the following relation \cite[p.224]{Lou2001}:
\[
\Spin_{1,3}^{0}\subset\Spin_{2,4}^{0}\cong SU(2,2).
\]
Fix such a Hermitian inner product $\langle\cdot|\cdot\rangle_{\rD}$,
and call it the \termi{Dirac form} or \termi{Dirac inner product}
.
\begin{prop}
\label{prop:complexConjugate}There exists an antilinear map $\cc:\C^{4}\to\C^{4}$
such that (1) $\cc^{2}=\id$ holds, and (2) $\cc$ is $\Spin_{1,3}^{0}$-covariant,
i.e., $\cc(g\psi)=g\cc(\psi)$ for all $g\in\Spin_{1,3}^{0}$ and
$\psi\in\C^{4}$.
\end{prop}

This proposition follows from Proposition \ref{prop:gamma-uniqueness}
and the existence of the Majorana representation, presented below.

If $\C^{4}$ is given a fixed Dirac inner product $\langle\cdot|\cdot\rangle_{\rD}$
and an antilinear map $\cc$ of Proposition \ref{prop:complexConjugate},
then $\SYM{\DSpinors_{4}}{C4}:=(\C^{4},\langle\cdot|\cdot\rangle_{\rD},\cc)$
is called a \termi{space of Dirac 4-spinors}. Note that $\C^{4}$
may be replaced by an abstract complex linear space $V$ with $\dim_{\C}V=4$;
in other words, we assume that $\C^{4}$ is given the usual complex
linear structure, but does not have any fixed positive-definite inner
product. On the other hand, we always assume that $\Cl_{1,3}$ (and
$\Spin_{1,3}$) acts on $\DSpinors_{4}$, i.e., that $\DSpinors_{4}$
is a $\Cl_{1,3}$ -module. Hence we may think that this is a part
of the definition of $\DSpinors_{4}$.

Let $\DSpinors_{4}^{*}$ be the dual space of $\DSpinors_{4}$. Any
$f\in\DSpinors_{4}^{*}$ corresponds to a unique $\psi\in\DSpinors_{4}$
by $f(\psi')=\langle\psi|\psi'\rangle_{\rD}$ for all $\psi'\in\DSpinors_{4}$.
In this case we write $f=\SYM{\Dconj}{\Dconj}\psi\equiv\psi^{\Dconj}$,
and called it the \termi{Dirac conjugate} of $\psi$.\footnote{Of course, this is a non-standard notation. In the physics literature,
it is denoted by $\ol{\psi}$, which can easily confused with the
complex conjugate. Some authors use $\psi^{+}$, $\psi^{\dagger}$,
$A\psi$, etc. \cite{San2008,DHP2009,BD2015}, but also these notations
seem confusing.} The map $\Dconj:\DSpinors_{4}\to\DSpinors_{4}^{*}$ is antilinear.
Define the Dirac inner product $\langle\cdot|\cdot\rangle_{\rD}$
on $\DSpinors_{4}^{*}$ by $\SYM{\langle u^{\Dconj}|v^{\Dconj}\rangle_{\rD}}{<d>D}:=\langle v|u\rangle_{\rD}$,
$u,v\in\DSpinors_{4}$. In the literature, it is customary to express
$v\in\DSpinors_{4}$ and $v^{\Dconj}\in\DSpinors_{4}^{*}$ by column
and row vectors, respectively. However this is not well consistent
with the notation $\langle u^{\Dconj}|v^{\Dconj}\rangle_{\rD}$, and
so we do not adopt this convention.

The \termi{Weyl (or chiral) representation} of $\Cl_{1,3}$ is given
by
\[
\gamma_{0}=\begin{pmatrix}0 & \sigma_{0}\\
\sigma_{0} & 0
\end{pmatrix},\quad\gamma_{i}=\begin{pmatrix}0 & -\sigma_{i}\\
\sigma_{i} & 0
\end{pmatrix},\qquad\sigma_{0}:=\begin{pmatrix}1 & 0\\
0 & 1
\end{pmatrix},\quad i=1,2,3
\]
where $\sigma_{i}$ are the Pauli matrices:
\[
\sigma_{1}:=\begin{pmatrix}0 & 1\\
1 & 0
\end{pmatrix},\quad\sigma_{2}:=\begin{pmatrix}0 & -\im\\
\im & 0
\end{pmatrix},\quad\sigma_{3}:=\begin{pmatrix}1 & 0\\
0 & -1
\end{pmatrix}.
\]
In the physics literature, including \cite{San2008,DHP2009,BD2015},
the \termi{Dirac conjugation matrix} is frequently used. It is defined
to be $A\in\Mat(4,\C)$ such that
\[
\langle\psi_{1}|\psi_{2}\rangle_{\rD}=\langle\psi_{1}|A\psi_{2}\rangle,\qquad\text{for all }\psi_{1},\psi_{2}\in\DSpinors_{4}
\]
where $\langle\cdot|\cdot\rangle$ is the usual inner product on $\C^{4}$.
In the Weyl representation, we can set $A=\gamma_{0}.$ However we
will avoid to refer to the Dirac conjugation matrix $A$ (and $\langle\cdot|\cdot\rangle$
on $\C^{4}$) in the subsequent subsections.

If $\gamma_{a}$'s are represented so that each matrix element is
pure imaginary, this is called a \termi{Majorana representation}
(see e.g., \cite{Pal2011}). In this case, the Dirac equation $(\pm\im\gamma^{a}\di_{a}+m)\psi=0$
in Minkowski space has real solutions like the Klein\textendash Gordon
equation. An example of the Majorana representation is given as follows:
\[
\gamma_{0}=\begin{pmatrix}0 & \sigma_{2}\\
\sigma_{2} & 0
\end{pmatrix},\qquad\gamma_{1}=\begin{pmatrix}\im\sigma_{1} & 0\\
0 & \im\sigma_{1}
\end{pmatrix},\qquad\gamma_{2}=\begin{pmatrix}0 & \sigma_{2}\\
-\sigma_{2} & 0
\end{pmatrix},\qquad\gamma_{3}=\begin{pmatrix}\im\sigma_{3} & 0\\
0 & \im\sigma_{3}
\end{pmatrix}.
\]
We see that each matrix element of $\gamma_{a}\gamma_{b}$, $a,b=0,...,3$,
is real. Thus this gives a \emph{real} representation of $\Spin_{1,3}^{0}$
on $\R^{4}$.

Let $\gamma_{a}^{{\rm Weyl}}$ denote the Weyl representations, and
$\gamma_{a}^{{\rm Maj}}$ the above Majorana representation of $\gamma_{a}$.
Then we can check%
\[
\gamma_{a}^{{\rm Maj}}=U^{-1}\gamma_{a}^{{\rm Weyl}}U,\qquad U:=\frac{1}{2}\begin{pmatrix}\im & -1 & -1 & \im\\
1 & \im & -\im & -1\\
-1 & \im & -\im & 1\\
-\im & -1 & -1 & -\im
\end{pmatrix}.
\]
The Dirac conjugation matrix in this Majorana representation is given
by
\[
A^{{\rm Maj}}=U^{-1}A^{{\rm Weyl}}U=\gamma_{0}^{{\rm Maj}}.
\]
Note that for any $\psi,\psi_{1},\psi_{2}\in\R^{4}$, $\langle\psi_{1}|\psi_{2}\rangle_{\rD}$
is pure imaginary, and $\langle\psi|\psi\rangle_{\rD}=0$ in the Majorana
representation.

Let us use the non-standard notation $\cconj\psi\equiv\psi^{\cconj}$
for the complex conjugate of $\psi\in\C^{4}$. In a Majorana representation,
the antilinear map $\cc=\cc^{{\rm Maj}}:=\cconj$ satisfies the condition
in Proposition \ref{prop:complexConjugate}. Let $\SYM{\RDSPinors_{4}}{R4}$
denote the invariant subspace of $\DSpinors_{4}$ w.r.t.~$\cc$.
In the case of the Majorana representation, we have $\RDSPinors_{4}=(\R^{4},\langle\cdot|\cdot\rangle_{\rD},\cc^{{\rm Maj}})$,
but here $\cc^{{\rm Maj}}$ is redundant since $\cc^{{\rm Maj}}\upha\RDSPinors_{4}=\id$.
Note that this $\langle\cdot|\cdot\rangle_{\rD}$ is not a $\R$-valued
bilinear form on $\RDSPinors_{4}$, but an $\im\R$-valued one. Roughly
speaking, $\psi^{\cconj}$ corresponds to the charge conjugate of
$\psi$ in a Majorana representation.

For a complex matrix $X$, let $X^{\cconj}$ denote the elementwise
complex conjugate of $X$. In the Weyl representation, the corresponding
conjugation operator is given by $\cc^{{\rm Weyl}}:=U\circ\cc^{{\rm Maj}}\circ U^{-1}=U\circ\cconj\circ U^{-1}$.
More explicitly, it is expressed as
\[
\cc^{{\rm Weyl}}=\im\gamma_{2}^{{\rm Weyl}}\circ\cconj.
\]
In fact, since $U\left(U^{-1}\right)^{\cconj}=\im\gamma_{2}^{{\rm Weyl}}$,
we have
\[
\cc^{{\rm Weyl}}\psi=U\cconj U^{-1}\psi=U\left(U^{-1}\right)^{\cconj}\psi^{\cconj}=\im\gamma_{2}^{{\rm Weyl}}\psi^{\cconj}.
\]

Define the real four-dimensional vector subspace $\RDSPinors^{{\rm Weyl}}$
of $\C^{4}$ by
\[
\SYM{\RDSPinors_{4}^{{\rm Weyl}}}{RWeyl}:=\{\psi\in\C^{4}|\cc^{{\rm Weyl}}\psi=\psi\}=U\R^{4}.
\]
Then $\RDSPinors_{4}^{{\rm Weyl}}$ is an invariant $\R$-subspace
of the representation $\rho_{{\rm Weyl}}$ of $\Spin_{1,3}^{0}$ on
$\C^{4}$ obtained from the Weyl representation $\pi_{{\rm Weyl}}$
of $\Cl_{1,3}$, and hence the restriction of $\rho_{{\rm Weyl}}$
onto $\RDSPinors_{4}^{{\rm Weyl}}$ gives a real four-dimensional
representation of $\Spin_{1,3}^{0}$.

\subsection{The Dirac/Majorana field in Minkowski space}

\label{subsec:Dirac-field-M}

In this subsection, we recall the algebraic structures of the Dirac/Majorana
field on a Minkowski spacetime.%
{} In the next subsection, the Dirac/Majorana field in curved spacetime
is defined to be a natural generalization of it in the Minkowski case.

A function $\psi:\R^{1,3}\to\DSpinors_{4}$ satisfying the Dirac equation
$(\im\gamma^{\alpha}\di_{\alpha}+m)\psi=0$ is sometimes called a
classical Dirac field on the Minkowski space $\R^{1,3}$, while of
course it concerns the relativistic \emph{quantum} mechanics. Consider
the quantization $\hat{\psi}$ of classical Dirac field $\psi$, but
simply written as $\psi$ in this subsection. The Dirac conjugate
of $\psi$ is conventionally denoted similarly to the classical case
(e.g., as $\bar{\psi}$ in the physics literature). However, we denote
it by $\phi$ in this subsection.

For simplicity, we fix a Majorana representation of gamma matrices.
Let $\cH$ be a Hilbert space, and $\Bdd(\cH)$ the set of bounded
operators on $\cH$. Let $\psi_{\alpha}(x)$ and $\phi_{\alpha}(x)$
be $\Bdd(\cH)$-valued distributions on $\R^{1,3}$, for each $\alpha=1,...,4$.
(More precisely, $\psi_{\alpha}$ and $\phi_{\alpha}$ are $\Bdd(\cH)$-valued
linear functionals on the test function space $C_{\cpt}^{\infty}(\R^{1,3},\bbK)$,
where $\bbK=\R$ or $\C$.) Let $\psi(x):=(\psi_{1}(x),...,\psi_{4}(x))$,
$\phi(x):=(\phi_{1}(x),...,\phi_{4}(x))$; We assume that for any
$\xi\in\cH$, $\langle\xi|\psi(x)\xi\rangle$ (resp.~$\langle\xi|\phi(x)\xi\rangle$)
is a $\DSpinors_{4}$-valued (resp.~$\DSpinors_{4}^{*}$-valued)
distribution.%

The causal propagator is defined by
\[
\SYM{\Delta_{m}(x)}{Delta}:=\frac{1}{(2\pi)^{3}\im}\int d^{4}p\ {\rm sgn}(p_{0})\delta(p^{2}-m^{2})e^{-\im px},
\]
and
\[
\SYM{\GreenOp^{m}(x)}{Sm}:=\left(\im\gamma^{\mu}\di_{\mu}+m\right)\Delta_{m}(x).
\]
Note that each matrix element of $\GreenOp^{m}(x)$ is a $\R$-valued
distribution in the Majorana representation. The anticommutation relations
for the Dirac field is given by
\begin{equation}
\{\psi_{\alpha}(x),\phi_{\beta}(y)\}=\im\GreenOp_{\alpha\beta}^{m}(x-y),\quad\{\psi_{\alpha}(x),\psi_{\beta}(y)\}=\{\phi_{\alpha}(x),\phi_{\beta}(y)\}=0,\qquad\alpha,\beta=1,...,4.\label{eq:anticomm-DiracM}
\end{equation}
For $u\in\DSpinors_{4}$, define $\psi_{u}$ by%
\[
\SYM{\psi_{u}(x)}{psiu}:=\langle u|\psi(x)\rangle_{\rD}.
\]
This notation is understood as follows. For $\xi\in\cH$, define the
$\DSpinors_{4}$-valued distribution $\psi^{\xi}$ by
\[
\SYM{\psi^{\xi}(x)}{psixi}:=(\psi_{1}^{\xi}(x),...,\psi_{4}^{\xi}(x)),\qquad\psi_{\alpha}^{\xi}(x):=\langle\xi|\psi_{\alpha}(x)\xi\rangle.
\]
Define the $\C$-valued distribution $\psi_{u}^{\xi}$ by
\[
\psi_{u}^{\xi}(x):=\langle u|\psi^{\xi}(x)\rangle_{\rD}.
\]
Then the $\Bdd(\cH)$-valued distribution $\psi_{u}$ is defined by
\[
\langle\xi|\psi_{u}(x)\xi\rangle=\psi_{u}^{\xi}(x),\qquad\text{for all }\xi\in\cH.
\]
Similarly, define the $\Bdd(\cH)$-valued distribution $\phi_{u^{\Dconj}}$
by $\phi_{u^{\Dconj}}(x):=\langle u^{\Dconj}|\phi(x)\rangle_{\rD}.$

The Dirac field satisfies
\begin{equation}
\psi_{u}^{*}(x)=\phi_{u^{\Dconj}}(x).\label{eq:psi*=00003Dphi}
\end{equation}
The algebraic structure of the Dirac field is completely characterized
by (\ref{eq:anticomm-DiracM}) and (\ref{eq:psi*=00003Dphi}). However,
$\phi(x)$ is algebraically redundant by (\ref{eq:psi*=00003Dphi}),
and instead we consider the anticommutation relations for $\psi$
and $\psi^{*}$:
\begin{equation}
\{\psi_{u}(x),\psi_{v}^{*}(y)\}=\im\GreenOp_{uv}^{m}(x-y),\quad\{\psi_{u}(x),\psi_{v}(y)\}=0,\qquad u,v\in\DSpinors_{4},\label{eq:CAR-psipsi-M}
\end{equation}
where
\[
\GreenOp_{uv}^{m}(x):=\DiracOpm_{m,u,v}\Delta_{m}(x).
\]
\[
\DiracOpm_{m,u,v}:=\langle u|\left(\im\gamma^{\mu}\di_{\mu}+m\right)v\rangle_{\rD}=\im\langle u|\gamma^{\mu}v\rangle_{\rD}\di_{\mu}+m\langle u|v\rangle_{\rD}.
\]
Note that there is no loss of information even if $\DSpinors_{4}$
is replaced by $\RDSPinors_{4}$ in (\ref{eq:CAR-psipsi-M}). Also
notice that $\langle u|v\rangle_{\rD}\in\im\R$ and $\langle u|\gamma^{\mu}v\rangle_{\rD}\in\R$
for all $u,v\in\RDSPinors_{4}$, and hence each matrix element of
$\im\GreenOp_{uv}^{m}$ is a $\R$-valued distribution in this case.%

{} Let $\SYM{C_{\cpt}^{\infty}(X,Y)}{Ccinf}$ denote the space of compactly
supported smooth functions from $X$ to $Y$. For $h\in C_{\cpt}^{\infty}(\R^{1,3},\R)$
and $u\in\DSpinors_{4}$, let $\psi(hu):=\int_{\R^{1,3}}\d x\,h(x)\psi_{u}(x)$.
For any $f\in C_{\cpt}^{\infty}(\R^{1,3},\DSpinors_{4})$, $\psi(f)$
is defined by extending the above $\psi(hu)$, $\C$-antilinearly
and continuously.

The integral kernel $\GreenOp^{m}(x-y)$ defines an operator on $C_{\cpt}^{\infty}(\R^{1,3},\DSpinors_{4})$,
and the sesquilinear form
\[
\GreenOp^{m}(f_{1},f_{2}):=\int_{\R^{1,3}}\d x\,\langle f_{1}(x)|(\GreenOp^{m}f_{2})(x)\rangle_{\rD},\qquad f_{1},f_{2}\in C_{\cpt}^{\infty}(\R^{1,3},\DSpinors_{4}).
\]
Let $\SYM{\langle f_{1}|f_{2}\rangle_{m}}{<>m}:=\im\GreenOp^{m}(f_{1},f_{2})$,
then $\langle\cdot|\cdot\rangle_{m}$ turns out to be a positive-semidefinite
Hermitian form on $C_{\cpt}^{\infty}(\R^{1,3},\DSpinors_{4})$. Then
the anticommutation relations (\ref{eq:CAR-psipsi-M}) are restated
as the following relation, which is accordant with the complex version
of the CAR algebra (Definition \ref{def:CAR-complex}):
\begin{equation}
\left\{ \psi(f_{1}),\psi(f_{2})^{*}\right\} =\langle f_{1}|f_{2}\rangle_{m}\qquad\left\{ \psi(f_{1}),\psi(f_{2})\right\} =0,\qquad f_{1},f_{2}\in C_{\cpt}^{\infty}(\R^{1,3},\DSpinors_{4}).\label{eq:CAR-psipsi*-M}
\end{equation}
Next, let
\[
\SYM{\Psi(f)}{Psi}:=\psi(f)+\psi^{*}(f),\qquad f\in C_{\cpt}^{\infty}(\R^{1,3},\DSpinors_{4}).
\]
Then we have the real version of the CAR (Definition \ref{def:CAR-real}):
\begin{equation}
\left\{ \Psi(f_{1}),\Psi(f_{2})\right\} =2(f_{1}|f_{2})_{m},\qquad f_{1},f_{2}\in C_{\cpt}^{\infty}(\R^{1,3},\DSpinors_{4}),\label{eq:CAR-PsiPsi-M}
\end{equation}
where $(f_{1}|f_{2})_{m}:=\Re\langle f_{1}|f_{2}\rangle_{m}.$ Therefore,
the CAR algebra for the Dirac field is given by
\[
\cA_{{\rm Dirac}}:=\CAR(C_{\cpt}^{\infty}(\R^{1,3},\DSpinors_{4}),(\cdot|\cdot)_{m})\cong\CAR(C_{\cpt}^{\infty}(\R^{1,3},\DSpinors_{4}),\langle\cdot|\cdot\rangle_{m}).
\]
If we replace $\DSpinors_{4}$ with $\RDSPinors_{4}$ in (\ref{eq:CAR-PsiPsi-M}),
we obtain the CAR algebra for the Majorana field:
\[
\cA_{{\rm Majorana}}:=\CAR(C_{\cpt}^{\infty}(\R^{1,3},\RDSPinors_{4}),(\cdot|\cdot)_{m}).
\]
{} Since $C_{\cpt}^{\infty}(\R^{1,3},\RDSPinors_{4})$ seems to have
no canonical complex structure, $\cA_{{\rm Majorana}}$ should be
defined by Definition \ref{def:CAR-real} (real version). On the other
hand, the Dirac field has the $U(1)$ gauge symmetry, which determines
a canonical complex structure. However note that this implies that
the complex structure $\ComplexS$ of the Dirac field is a gauge transformation,
and hence its empirical meaning of is far from clear (or it is empirically
meaningless). Generally, it is difficult to find the \emph{empirical}
ground for the ``canonicalness'' of gauge symmetries in quantum
physics (see e.g., Haag \cite{Haa96} and Araki \cite{Araki99}).

\subsection{Quantum fields in curved spacetime}

Let $M=(\mathcal{M},g,SM,p)$ be a globally hyperbolic spin spacetime,
which is given the canonical connection determined by the Levi-Civita
connection on $TM$ \cite{San2008,DHP2009,BD2015}.%
{} The \termi{Dirac spinor bundle} $DM$ is defined to be the associated
vector bundle
\[
\SYM{DM}{DM}:=SM\times_{\Spin_{1,3}^{0}}\DSpinors_{4}.
\]
We denote the dual vector bundle of $DM$ by $D^{*}M$, and call it
the \termi{Dirac cospinor bundle}.

We saw that there exists a $\Spin_{1,3}^{0}$-invariant subspace $\RDSPinors_{4}\subset\DSpinors_{4}$,
$\RDSPinors_{4}\cong\R^{4}$; If $\pi$ is a Majorana representation,
this is simply given by $\RDSPinors_{4}=\R^{4}\subset\C^{4}$. Hence
the complex vector bundle $DM$ has the ``real part subbundle''
\[
\SYM{D_{\re}M}{DM}:=SM\times_{\Spin_{1,3}^{0}}\RDSPinors_{4}.
\]
We call $D_{\re}M$ the \termi{Majorana spinor bundle}. If we fix
a Majorana representation of gamma matrices, we can set $\RDSPinors_{4}=\R^{4}$.

Generally, for a vector bundle $B$ over $M$, let $\Gamma^{\infty}(B)$
denote the space of smooth sections of $B$, and $\Gamma_{\cpt}^{\infty}(B)$
the compactly supported smooth section of $B$.

The \termi{Dirac operator} $\,\SYM{\nablaslash}{nablabar}:\!\Gamma^{\infty}(DM)\!\rightarrow\!\Gamma^{\infty}(DM)$
is the differential operator defined by $\nablaslash:=\gamma^{\mu}\nabla_{\mu}$.
Dually, we define the Dirac operator $\nablaslash:\!\Gamma^{\infty}(D^{*}M)\!\rightarrow\!\Gamma^{\infty}(D^{*}M)$
on the cospinor bundle.

Fix $m\ge0$. Define the differential operators $\DiracOpm_{m}:\Gamma^{\infty}(DM)\to\Gamma^{\infty}(DM)$
and $\tilde{\DiracOpm}_{m}:\Gamma^{\infty}(D^{*}M)\to\Gamma^{\infty}(D^{*}M)$
by
\[
\SYM{\DiracOpm_{m}}D:=-\im\nablaslash+m,\qquad\SYM{\tilde{\DiracOpm}_{m}}D:=\im\nablaslash+m.
\]
The \termi{Dirac equation} for $u\in\Gamma^{\infty}(DM)$, respectively
for $v\in\Gamma^{\infty}(D^{*}M)$, is given by
\begin{equation}
\DiracOpm_{m}u=0,\qquad\tilde{\DiracOpm}_{m}v=0.\label{Diraceqn}
\end{equation}

For $u_{1},u_{2}\in\Gamma^{\infty}(DM)$ (or $\Gamma^{\infty}(D^{*}M)$),
let $\langle u_{1}|u_{2}\rangle_{\rD}\in C^{\infty}(M)$ denote the
pointwise Dirac inner product. Define the indefinite Hermitian inner
product $\SYM{\langle\cdot|\cdot\rangle_{\rD,M}}{<>Dm}$ on $\Gamma_{{\rm c}}^{\infty}(DM)$
by
\[
\langle u|v\rangle_{\rD,M}:=\int_{M}\langle u|v\rangle_{\rD}\,\d\mathrm{vol}_{g}.
\]
More generally, if one of $u,v\in\Gamma^{\infty}(DM)$ is compactly
supported, the integral on the r.h.s.~converges, and hence $\langle u|v\rangle_{\rD,M}$
is defined.

\begin{thm}
There exist the (unique) advanced ($-$) and retarded ($+$) Green
operators $\SYM{\GreenOp^{\pm}}{S+-}=\GreenOp_{m}^{\pm}$ for $\DiracOpm_{m}$,
that is, the operators ${\SYM{\GreenOp_{m}^{\pm}}{S+-m}}:{\Gamma_{{\rm c}}^{\infty}(DM)}\to{\Gamma^{\infty}(DM)}$
such that $\DiracOpm_{m}\GreenOp_{m}^{\pm}=\GreenOp_{m}^{\pm}\DiracOpm_{m}=\id$
and
\[
\supp(\GreenOp_{m}^{\pm}f)\subset J^{\pm}(\supp\,f),\qquad\text{for all }f\in\Gamma_{{\rm c}}^{\infty}(DM),
\]
where $J^{\pm}(U)$ denotes the causal future/past of the set $U$.

Furthermore, they satisfy $\langle f_{1}|\GreenOp_{m}^{\pm}f_{2}\rangle_{\rD,M}=\langle\GreenOp_{m}^{\mp}f_{1}|f_{2}\rangle_{\rD,M}$
for all $f_{1},f_{2}\in\Gamma_{{\rm c}}^{\infty}(DM)$.

Dually, the Green operators on the Dirac cospinor bundle ${\SYM{\tilde{\GreenOp}_{m}^{\pm}}{S+-}}:{\Gamma_{{\rm c}}^{\infty}(D^{*}M)}\to{\Gamma^{\infty}(D^{*}M)}$
exist.
\end{thm}

Let $\SYM{\GreenOp}S\equiv\SYM{\GreenOp_{m}}{Sm}:=\GreenOp_{m}^{-}-\GreenOp_{m}^{+}$,
sometimes called the \termi{causal propagator} for $\DiracOpm_{m}$.%
{} Let
\[
\SYM{\langle f_{1}|f_{2}\rangle_{m}}{(f,h)}:=\im\langle f_{1}|\GreenOp_{m}f_{2}\rangle_{\rD,M},\qquad f_{1},f_{2}\in\Gamma_{{\rm c}}^{\infty}(DM).
\]

\begin{thm}
\label{positivity} %
\begin{enumerate}
\item $\langle\cdot|\cdot\rangle_{m}$ is a positive-semidefinite Hermitian
form on $\Gamma_{{\rm c}}^{\infty}(DM)$.
\item $\langle f|f\rangle_{m}=0$ iff $f\in\ker\GreenOp_{m}$. %
\end{enumerate}
\end{thm}

We require that the CAR (\ref{eq:CAR-psipsi*-M}) for the Dirac field
in a Minkowski spacetime is directly generalized to a curved spacetime,
that is,
\[
\left\{ \psi(f_{1}),\psi(f_{2})^{*}\right\} =\langle f_{1}|f_{2}\rangle_{m},\qquad f_{1},f_{2}\in\Gamma_{\cpt}^{\infty}(DM).
\]
Let
\[
\Psi(f):=\psi(f)+\psi(f)^{*},\qquad f\in\Gamma_{\cpt}^{\infty}(DM),
\]
then we have the real version of the CAR:
\begin{equation}
\left\{ \Psi(f_{1}),\Psi(f_{2})\right\} =2(f_{1}|f_{2})_{m},\qquad f_{1},f_{2}\in\Gamma_{\cpt}^{\infty}(DM),\label{eq:CAR-PsiPsi-D-cur}
\end{equation}
where $(f_{1}|f_{2})_{m}:=\Re\langle f_{1}|f_{2}\rangle_{m}.$ Thus,
the algebra for the Dirac field is defined to be
\[
\cA_{{\rm Dirac}}:=\CAR(\Gamma_{\cpt}^{\infty}(DM),(\cdot|\cdot)_{m}).
\]
On the other hand, the algebra for the Majorana field is given by
\[
\cA_{{\rm Majorana}}:=\CAR(\Gamma_{\cpt}^{\infty}(D_{\re}M),(\cdot|\cdot)_{m}).
\]

\section{The simplest examples of empirical laws in a CAR algebra}

\label{sec:Empirical-law-CAR}

Let $(V,(\cdot|\cdot))$ be a real pre-inner product space, and consider
the empirical laws of a quantum system described by the CAR algebra
$\CAR(V)=\CAR(V,(\cdot|\cdot))$, such as the Majorana field in curved
spacetime. Of course, it is meaningless to talk about the empirical
law of the purely mathematical structure $\CAR(V)$ itself, without
physical/empirical meanings. Hence we make some physical interpretations
of $\CAR(V)$. Each operator $\phi(v)\in\CAR(V)$, $v\in V$, is interpreted
as (a constituent of) a fermion field operator. A fermion field operator
is not an observable; A fermion field has the $\Z_{2}$-gauge symmetry
$\phi(v)\mapsto-\phi(v)$. The product of two field operators $\phi(v_{1})\phi(v_{2})$
is $\Z_{2}$-gauge invariant, and more generally, the subalgebra $\CAR(V)_{{\rm even}}$
of $\CAR(V)$ generated by $\{\phi(v_{1})\phi(v_{2})|v_{1},v_{2}\in V\}$
consists of $\Z_{2}$-invariant elements of $\CAR(V)$. Assume that
every selfadjoint element of $\CAR(V)_{{\rm even}}$ is an observable,
in other words, that $\Z_{2}$ is the only gauge symmetry of the system.

Let $(W,(\cdot|\cdot))$ be a pre-inner product space, and let $\cN:=\{v\in W|\,(v|v)=0\}$.
Then $V:=W/\cN$ is an inner product space.%
{} Thus, without loss of generality, we can assume that $(\cdot|\cdot)$
is an inner product on $V$.

\subsection{Finite-dimensional case}

First, assume that $(V,(\cdot|\cdot))$ is a \emph{finite-dimensional}
real inner product space.%
{} In this case, furthermore we may assume that $\dim V$ is even without
loss of generality, since $\CAR(\R^{2n})\cong\Mat(2^{n},\C)$, and
\[
\CAR(\R^{2n+1})\cong\Mat(2^{n},\C)\oplus\Mat(2^{n},\C)\cong\CAR(\R^{2n})\oplus\CAR(\R^{2n}),\qquad n\in\N.
\]

Let $V\cong\R^{2n}$. Let $\Tr_{V}$ be the canonical trace on $\CAR(V)$,
i.e., the usual trace on $\Mat(2^{n},\C)\cong\CAR(V)$, so that $\Tr_{V}\bbOne=2^{n}$.

For $k\in\N$, let $\SYM{\ON_{k}(V)}{ONk}$ denote the set of orthonormal
system of $V$ of cardinal $k$, especially
\begin{equation}
\ON_{2}(V):=\{(u,v)\in V\times V;\,\|u\|=\|v\|=1,\ (u|v)=0\}.\label{eq:def:ON2}
\end{equation}
For $(u,v)\in\ON_{2}(V)$, let
\begin{equation}
\SYM{X_{u,v}}{Xuv}:=\im\phi(u)\phi(v).\label{eq:def:Xuv}
\end{equation}
Then $X_{u,v}$ is selfadjoint and $X_{u,v}^{2}=\bbOne$. Hence
\begin{equation}
\SYM{P_{u,v}^{\pm}}{Puv+-}:=\frac{1}{2}\left(\bbOne\pm X_{u,v}\right)\label{eq:def:Puv+-}
\end{equation}
is the spectral projections of $X_{u,v}$, which is observable. Since
$P_{u,v}^{+}=P_{v,u}^{-}$, we may consider only $P_{u,v}^{+}$.

For $(u,v),(u',v')\in\ON_{2}(V)$, let
\begin{equation}
\Prob_{V}(u',v'|u,v)\equiv\Prob_{V}\left(P_{u',v'}^{+}|P_{u,v}^{+}\right):=\frac{\Tr_{V}P_{u,v}^{+}P_{u',v'}^{+}}{\Tr_{V}P_{u,v}^{+}}.\label{eq:def:P(u'v'|uv)}
\end{equation}
$\Prob_{V}(u',v'|u,v)$ is interpreted as the \emph{a priori} probability
of $P_{u',v'}^{+}$ after the measurement of $P_{u,v}^{+}$. More
precisely, interpret $P_{u,v}^{+}$ and $P_{u',v'}^{+}$ as yes-no
type measurements. Then $\Prob_{V}(u',v'|u,v)$ is interpreted as
the \emph{a priori} probability that $P_{u',v'}^{+}$ is ``yes'',
after one observed that $P_{u,v}^{+}$ is ``yes''. We call this
type of probability a \termi{prior conditional probability}.

Here the spaces $V$ and $\ON_{2}(V)\subset V^{2}$ are assumed to
be physically meaningful ``parameter spaces''; $X_{u,v}$ and $P_{u,v}^{+}$
are observables parametrized by the physically meaningful parameter
$(u,v)\in\ON_{2}(V)$. This implies that a formula to give the value
of $\Prob(u',v'|u,v)$ may be said to be an example of the \emph{physical/empirical
laws,} and hence (\ref{eq:def:P(u'v'|uv)}) itself will be a kind
of empirical law. (Mathematically, (\ref{eq:def:P(u'v'|uv)}) is a
\emph{definition}, and so is not a (nontrivial) \emph{law}; the equation
does not contain ``$=$'', but ``$:=$''. However, the l.h.s.~of
(\ref{eq:def:P(u'v'|uv)}) has an intended empirical/experimental
meaning as above. Hence the equation (\ref{eq:def:P(u'v'|uv)}) is
not a pure definition, and it may be experimentally verified or falsified.)
However the r.h.s.~of (\ref{eq:def:P(u'v'|uv)}) is slightly too
implicit to be called a physical/empirical law; It is not immediately
clear how to calculate $\Prob(u',v'|u,v)$ by the r.h.s.~of (\ref{eq:def:P(u'v'|uv)}),
although this is computable in principle, because $P_{u,v}^{+}$ can
be represented as a finite-dimensional matrix, and so we can compute
all the matrix elements of $P_{u,v}^{+}$. However, for large $n$,
a $2^{n}\times2^{n}$ matrix is ``effectively infinite-dimensional'',
and is very hard to compute. On the other hand, one may expect that
$\Prob_{V}(u',v'|u,v)$ can be expressed as a function of the four
values $(u|u')$, $(u|v')$, $(v|u')$ and $(v|v')$ only, instead
of $2^{n}\times2^{n}$ matrix elements. In fact, we find the following
simple explicit expression for the r.h.s.~of (\ref{eq:def:P(u'v'|uv)}):
\begin{equation}
\Prob_{V}(u',v'|u,v)=\frac{1}{2}\left(1+(u|u')(v|v')-(u|v')(v|u')\right).\label{eq:P(u'v'|uv)-law}
\end{equation}
This expression is of a far more desirable form as an empirical law
than (\ref{eq:def:P(u'v'|uv)}); Eq.~(\ref{eq:P(u'v'|uv)-law}) is
not only easy to compute, but also shows the $SO(V)$-invariance of
$\Prob_{V}(u',v'|u,v)$ manifestly. Generally it is desirable that
an empirical law is expressed in a manifestly covariant manner, with
respect to the physical symmetries. Eq.~(\ref{eq:P(u'v'|uv)-law})
is derived from $\Tr_{V}P_{u,v}^{+}=2^{n-1}$ and the following
\begin{lem}
For any $(u,v),(u',v')\in\ON_{2}(V)$,
\[
\Tr_{V}X_{u,v}X_{u',v'}=2^{n}\left[(u|u')(v|v')-(u|v')(v|u')\right].
\]
\end{lem}

\begin{proof}
{} Let $\{e_{1},...,e_{2n}\}$ be an orthonormal basis of $V$, such
that $u=e_{1}$, $v=e_{2}$. Let $u=\sum_{i=1}^{2n}u_{i}e_{i}$, $v=\sum_{i=1}^{2n}v_{i}e_{i}$,
$u'=\sum_{i=1}^{2n}u_{i}'e_{i}$, $v'=\sum_{i=1}^{2n}v_{i}'e_{i}$
($u_{i},v_{i},u_{i}',v_{i}'\in\R$), but indeed $u_{i}=\delta_{i1}$,
$v_{i}=\delta_{i2}$. %
{} Let $\phi_{i}:=\phi(e_{i})$. We see
\[
\Tr_{V}\phi_{i}\phi_{j}\phi_{k}\phi_{l}=\begin{cases}
2^{n} & i=j,\ k=l\text{ or }i=l,\ j=k\\
-2^{n} & i=k,\ j=l,\ i\neq j\\
0 & \text{otherwise}.
\end{cases}
\]
Hence we have
\begin{align*}
\Tr_{V}X_{u,v}X_{u',v'} & =-\sum_{ijkl}^{2n}u_{i}v_{j}u_{k}'v_{l}'\Tr\phi_{i}\phi_{j}\phi_{k}\phi_{l}\\
 & =-\sum_{k,l}^{2n}u_{k}'v_{l}'\Tr\phi_{1}\phi_{2}\phi_{k}\phi_{l}\\
 & =2^{n}u_{1}'v_{2}'-2^{n}u_{2}'v_{1}'\\
 & =2^{n}(u|u')(v|v')-2^{n}(v|u')(u|v').
\end{align*}
\end{proof}

Let $N\in\N$ and $(u_{i},v_{i})\in\ON_{2}(V)$, $i=1,...,N$, and
$1\le k<N$. The prior conditional probability (\ref{eq:def:P(u'v'|uv)})
is generalized as follows:
\begin{equation}
\Prob_{V}(u_{k+1},v_{k+1};\cdots;u_{N},v_{N}|u_{1},v_{1};\cdots;u_{k},v_{k})\equiv\Prob_{V}\left(B|A\right):=\frac{\Tr_{V}(AB)^{*}AB}{\Tr_{V}A^{*}A},
\end{equation}
where
\[
A:=P_{u_{1},v_{1}}^{+}\cdots P_{u_{k},v_{k}}^{+},\qquad B:=P_{u_{k+1},v_{k+1}}^{+}\cdots P_{u_{N},v_{N}}^{+}.
\]
This may be said to be a kind of empirical law. However, probably
the corresponding generalization of the formula (\ref{eq:P(u'v'|uv)-law})
must be far more complicated for large $N$. Generally, it is desirable
that a \emph{fundamental} empirical law is expressed simply; If a
law is complicated, one may expect that it can be derived from simpler
ones.\footnote{Of course, there is no \emph{a priori} reason for the existence of
such a simpler law, but actually such simpler laws have been often
found; This will be a kind of ``meta-empirical law'' in theoretical
physics.}%
{} It suggests that such a complicated generalization of the formula
(\ref{eq:P(u'v'|uv)-law}) is a \emph{non-fundamental (derivative)}
empirical law, which should be derived from more fundamental ones.

Let $\{e_{1},...,e_{2n}\}$ be an orthonormal basis of $V$, and $P_{k}:=P_{e_{k},e_{k+n}}^{+}$,
$k=1,...,n$. Then $P_{1},...,P_{n}$ are pairwise commutative, and
$\prod_{k=1}^{n}P_{k}$ is a minimal nonzero projection in $\CAR(V)$.
A natural idea to use such minimal nonzero projections to describe
more fundamental empirical laws.

\subsection{Infinite-dimensional case}

Next consider the case where $\dim V=\infty$. Let $(u,v),(u',v')\in\ON_{2}(V)$,
and $W$ be an arbitrary finite even dimensional subspace of $V$
such that $u,v,u',v'\in W$. Consider the subalgebra $\CAR(W)\subset\CAR(V)$.
Similarly to (\ref{eq:def:P(u'v'|uv)}), for $(u,v),(u',v')\in\ON_{2}(W)$,
we define the prior conditional probability
\begin{equation}
\Prob_{W}(u',v'|u,v)\equiv\Prob_{W}\left(P_{u',v'}^{+}|P_{u,v}^{+}\right):=\frac{\Tr_{W}P_{u,v}^{+}P_{u',v'}^{+}}{\Tr_{W}P_{u,v}^{+}}.\label{eq:def:P(u'v'|uv)-W}
\end{equation}
By (\ref{eq:P(u'v'|uv)-law}), we have
\begin{equation}
\Prob_{W}(u',v'|u,v)=\frac{1}{2}\left(1+(u|u')(v|v')-(u|v')(v|u')\right).\label{eq:P(u'v'|uv)-law-W}
\end{equation}
Note that the r.h.s.~of (\ref{eq:def:P(u'v'|uv)-W}) is not defined
if $W$ is replaced by $V$, since $\Tr_{V}P_{u,v}^{+}=\infty$ and
$\Tr_{V}P_{u,v}^{+}P_{u',v'}^{+}=\infty$. However, the r.h.s.~of
(\ref{eq:P(u'v'|uv)-law-W}) is independent of $W$. Thus it is reasonable
to define $\Prob_{V}(u',v'|u,v)$ by
\begin{equation}
\Prob_{V}(u',v'|u,v):=\lim_{W\to V}\Prob_{W}(u',v'|u,v)=\text{r.h.s. of }(\ref{eq:P(u'v'|uv)-law-W}).\label{eq:def:P(u'v'|uv)-infty}
\end{equation}
We regard (\ref{eq:def:P(u'v'|uv)-infty}) as an example of the empirical
law as a prior conditional probability in the infinite-dimensional
case. (As mentioned above, although the equation (\ref{eq:def:P(u'v'|uv)-infty})
is of the form of a mathematical \emph{definition}, it can be interpreted
as a nontrivial empirical \emph{law}, which may be verified or falsified.)

Note that the idea to use the minimal nonzero projections to describe
the fundamental empirical law depends on $W$; If $W'$ is a finite
even dimensional subspace of $V$ such that $W\subsetneq W'$, the
minimal nonzero projections in $\CAR(W)$ are not minimal in $\CAR(W')$.
Nevertheless we consider this idea to be promising, and examine it
in the following.

\section{Empirical laws: a general consideration}

\label{sec:Empirical-laws-general}

In this section we consider how to describe the empirical laws generally,
but we confine ourselves to the finite-dimensional cases. We do not
consider \emph{all} of empirical laws, but the \emph{fundamental}
ones only, where the other empirical laws are expected to be derived
from the fundamental ones.

Let $d\in\N$, and consider a physical system described by the $C^{*}$-algebra
$\cA\cong\Mat(d,\C)$. Let $\cP(\cA)$ be the set of projection in
$\cA$. Consider the subset $\{P_{x}|x\in\Manifold\}$ of $\cP(\cA)$,
where $\Manifold$ is an index set, which can be interpreted as a
physically meaningful parameter space, so that each $P_{x}$ is given
an empirical meaning as a yes-no type measurement. Assume that $\cA$
is generated by $\{P_{x}|x\in\Manifold\}$. Here we assume that a
topology is given on $\Manifold$. Actually $\Manifold$ is assumed
to be a differential manifold in most cases, since many physical laws
are expressed as differential equations.

Since we want to describe only the fundamental empirical laws, it
is natural to assume that each $P_{x}$ is of rank $1$, i.e., a minimal
nonzero projection in $\cA$. However note that such a projection
may not exist if the $C^{*}$-algebra is infinite-dimensional.

For $N\in\N$ and $x_{1},...,x_{N}\in\Manifold$, let
\begin{equation}
\Prob(x_{k+1},...,x_{N}|x_{1},...,x_{k}):=\frac{\Tr(AB)^{*}AB}{\Tr A^{*}A},\label{eq:def:P(x|x)}
\end{equation}
\[
A:=P_{x_{1}}\cdots P_{x_{k}},\qquad B:=P_{x_{k+1}}\cdots P_{x_{N}},\qquad1\le k<N.
\]
when the denominator is not zero. As mentioned above, this can be
said to be a kind of empirical law. However, it is doubtful that the
other empirical laws are derived from this. In fact, we see
\[
\Tr A^{*}A=\prod_{j=1}^{k-1}\Tr P_{x_{j}}P_{x_{j+1}}=\prod_{j=1}^{k-1}\Prob(x_{j+1}|x_{j})
\]
and hence%
\[
\Prob(x_{k+1},...,x_{N}|x_{1},...,x_{k})=\prod_{j=k}^{N-1}\Prob(x_{j+1}|x_{j}).
\]
Thus the law (\ref{eq:def:P(x|x)}) indeed has only the information
for the case $N=2$.

Next assume that $\Manifold$ is given a finite measure $\mu$, and
let

\[
\SYM{\Vac(S)}{E(S)}\equiv\SYM{\Vac_{S}}{ES}:=\int_{S}P_{x}\,\d\mu(x)
\]
for $S\subset\Borel(\Manifold)$, the Borel subsets of $\Manifold$.
We assume that $\mu$ is normalized so that $\Vac(\Manifold)=\bbOne$.
Then $\Vac(\cdot)$ is called a POVM (positive operator valued measure).
Although it is not very clear that each $\Vac_{S}$ has an empirical
meaning, we assume that $\Vac_{S}$ can be interpreted as a yes-no
type measurement with some error; Every $P\in\cP(\cA)$ is interpreted
as a yes-no type measurement without error in the sense that $\Prob(P|P)=1$,
but now $\Prob(S|S)$ ($S\in\Borel(\Manifold)$) may be less than
$1$, where the prior conditional probability $\Prob(\cdot|\cdot)$
is defined similarly to (\ref{eq:def:P(x|x)}), that is, for $S_{1},...,S_{n}\in\Borel(\Manifold)$,
\begin{equation}
\Prob(S_{k+1},...,S_{N}|S_{1},...,S_{k})\equiv\Prob(B|A):=\frac{\Tr(AB)^{*}AB}{\Tr A^{*}A},\label{eq:def:P(S|S)}
\end{equation}
\[
A:=\Vac(S_{1})\cdots\Vac(S_{k}),\qquad B:=\Vac(S_{k+1})\cdots\Vac(S_{N}),\qquad1\le k<N,
\]
if $\Tr A^{*}A\neq0$, i.e., $A\neq0$. In the finite-dimensional
cases, we have $\Tr\bbOne=d<\infty$ and hence we can define the (non-conditional)
probability
\[
\Prob(S_{1},...,S_{N}):=\Prob(S_{1},...,S_{N}|\Manifold)=\frac{\Tr A^{*}A}{\Tr\Vac(\Manifold)}=\frac{\Tr A^{*}A}{d},\qquad A:=\Vac(S_{1})\cdots\Vac(S_{N}),
\]
so that
\[
\Prob(S_{k+1},...,S_{N}|S_{1},...,S_{k})=\frac{\Prob(S_{1},...,S_{N})}{\Prob(S_{1},...,S_{k})}.
\]

Eq.~(\ref{eq:def:P(S|S)}) is expected to have more information as
an empirical law than (\ref{eq:def:P(x|x)}), but this definition
will require some explanations. Let $\cA\cong\Mat(d,\C)$ be represented
on the Hilbert space $\cH\cong\C^{d}$. Let $\rho$ denote a density
matrix on $\cH$. We interpret the unnormalized state transformation
$\rho\mapsto\Vac(S)\rho\Vac(S)$ as an observation (possibly with
some error) that the $\Manifold$-valued physical quantity is in $S\subset\Manifold$.%
{} However in this interpretation, the probability of the observation
is not ${\rm Tr}(\rho\Vac(S))$ but ${\rm Tr}(\rho\Vac(S)^{2})$.
Hence this probability is not additive w.r.t.~$S$ in general, that
is, even if $S_{1},S_{2}\in\Borel(\Manifold)$ and $S_{1}\cap S_{2}=\emptyset$,
generally we have
\[
\Prob_{\rho}(S_{1}\cup S_{2})\neq\Prob_{\rho}(S_{1})+\Prob_{\rho}(S_{2}),\qquad\Prob_{\rho}(S):=\Tr(\rho\Vac(S)^{2}).
\]
One may feel that this interpretation seems strange, so we will explain
this as follows. Let $S\mapsto\tilde{\Vac}(S)$ be a Naimark dilation
of $S\mapsto\Vac(S)$; that is, $\tilde{\Vac}(\cdot)$ is a projection-valued
measure (PVM) on some extended Hilbert space $\cH':=\cH\oplus\cK$,
so that $E_{\cH}\tilde{\Vac}(S)\phi=\Vac(S)\phi$ for all $S\in\Borel(\Manifold)$
and $\phi\in\cH$, where $E_{\cH}$ is the orthogonal projection from
$\cH'$ onto $\cH$. If this POVM measurement is realized as the composition
of the PVM measurement $\tilde{\Vac}(S)$ and a projective measurement
$E_{\cH}$, on $\cH$, then the resulting joint probability is
\[
\Tr_{\cH'}(E_{\cH}\tilde{\Vac}(S)(\rho\oplus0)\tilde{\Vac}(S)E_{\cH})=\Tr_{\cH}(\Vac(S)\rho\Vac(S))=\Tr_{\cH}(\rho\Vac(S)^{2}),
\]
and hence not additive w.r.t.~$S$, but this fact causes no contradiction
or paradox. In the literature, often the unnormalized state transformation
$\rho\mapsto\Vac(S)^{1/2}\rho\Vac(S)^{1/2}$ is assumed instead of
$\rho\mapsto\Vac(S)\rho\Vac(S)$, to maintain the additivity of the
probability $\Prob_{\rho}(S):=\Tr(\rho\Vac(S))$. However, this assumption
seems to lack a physical ground, and somewhat rather strange.

For $S\in\Borel(\Manifold)$, again let $\Prob_{\rho}(S):=\Tr(\rho\Vac(S)^{2})$.
Let $X_{k}\in\Borel(\Manifold)$ ($k=1,...,K$, $K\in\N$) and $i\neq j\then X_{i}\cap X_{j}=\emptyset$.%
{} If each $\Vac(X_{i})$ is a projection, we have the usual (classical)
additivity $\Prob_{\rho}(\bigcup_{i=1}^{K}X_{i})=\sum_{i=1}^{K}\Prob_{\rho}(X_{i})$,
and so the computation of $\Prob_{\rho}(\bigcup_{i=1}^{K}X_{i})$
is reduced to that of each $\Prob_{\rho}(X_{i})$. However, if $\Vac(X_{i})$'s
are not projections, how can we reduce the computation of $\Prob_{\rho}(\bigcup_{i=1}^{K}X_{i})$?
The answer is given by (\ref{eq:Sorkin-S}) below. First, let $\SYM{\Bdd_{1}(\cH)}{B1()}$
denote the set of bounded operators $A$ on $\cH$ such that $\|A\|\le1$.
Then we find the following.
\begin{lem}
[Sorkin additivity]\label{lem:Sorkin0} For any $A\in\Bdd_{1}(\cH)$,
set $\SYM{\Prob_{\rho}(A)}{Prho(A)}:=\Tr A\rho A^{*}$. Let $\cH$
be a (finite or infinite-dimensional) Hilbert space. Let $K\in\N$,
and $A_{i}\in\Bdd_{1}(\cH)$ ($i=1,...,K$). Assume $A_{i}+A_{j}\in\Bdd_{1}(\cH)$
($i\neq j$), and $A_{1}+\cdots+A_{K}\in\Bdd_{1}(\cH)$. Then we have%
\begin{equation}
\Prob_{\rho}\left(A_{1}+\cdots+A_{K}\right)=\sum_{1\le i<j\le K}\Prob_{\rho}(A_{i}+A_{j})-(K-2)\sum_{i=1}^{K}\Prob_{\rho}(A_{i}).\label{eq:Sorkin-A}
\end{equation}
\end{lem}

While the proof is easy, it seems that the significance of (\ref{eq:Sorkin-A})
as a quantum-probabilistic law is emphasized for the first time in
1994 by Sorkin \cite{Sor94} (see also \cite{Tab2009,Bar-Mul-Udu-2014,Mue2021}).
Thus we refer the law (\ref{eq:Sorkin-A}) as the \termi{Sorkin additivity}.
Let $A_{k}:=\Vac(X_{k})$, $k=1,...,K$, then we see that they satisfy
the assumptions of Lemma \ref{lem:Sorkin0}. Hence we have the following
quantum-probabilistic law in our setting:
\begin{equation}
\Prob_{\rho}\left(\bigcup_{i=1}^{K}X_{i}\right)=\sum_{1\le i<j\le K}\Prob_{\rho}(X_{i}\cup X_{j})-(K-2)\sum_{i=1}^{K}\Prob_{\rho}(X_{i}).\label{eq:Sorkin-S}
\end{equation}
We can replace $\Prob_{\rho}(\cdot)$ in (\ref{eq:Sorkin-S}) with
$\Prob(\cdot|S_{1},...,S_{k})$ ($S_{j}\in\Borel(\Manifold)$) defined
by (\ref{eq:def:P(S|S)}); More generally, we have
\begin{cor}
We write $\Prob_{A}(B):=\Prob(B|A)$ in (\ref{eq:def:P(S|S)}). Let
$S_{j}\in\Borel(\Manifold)$ for $j=1,...,k$, and $S_{j}^{(i)}\in\Borel(\Manifold)$
for $i=1,...,K$, $j=k+1,...,N$. Let $A:=\Vac(S_{1})\cdots\Vac(S_{k}),$
and $B_{i}:=\Vac(S_{k+1}^{(i)})\cdots\Vac(S_{N}^{(i)})$. Assume $B_{i}+B_{j}\in{}_{1}(\cH)$
($i\neq j$), and $B_{1}+\cdots+B_{K}\in{}_{1}(\cH)$. Then
\begin{equation}
\Prob_{A}\left(\sum_{i=1}^{K}B_{i}\right)=\sum_{1\le i<j\le K}\Prob_{A}(B_{i}+B_{j})-(K-2)\sum_{i=1}^{K}\Prob_{A}(B_{i}).\label{eq:def:Sorkin-B}
\end{equation}
\end{cor}

However the above Corollary is not very convenient. %
{} Instead we want the ``integral version'' of the Sorkin additivity,
where the summations are replaced by integrals. This is given by the
Sorkin density functions in Proposition \ref{prop:SorkinDensity}
below. %
First notice the following relation, which is analogous to (\ref{eq:Sorkin-A}):
\begin{lem}
Let $\mu$ be a measure on a set $X$. Let $f:X\to\C$ be measurable.
If $S\subset X$ is measurable and $\mu(S)<\infty$, we have
\begin{align*}
\left|\int_{S}f(x)\d\mu(x)\right|^{2} & =\frac{1}{2}\int_{S}\d\mu(x)\int_{S}\d\mu(y)\left|f(x)+f(y)\right|^{2}-\mu(S)\int_{S}\left|f(x)\right|^{2}\d\mu(x)\\
 & =\frac{1}{2}\int_{S}\d\mu(x)\int_{S}\d\mu(y)\left(\left|f(x)+f(y)\right|^{2}-\left|f(x)\right|^{2}-\left|f(y)\right|^{2}\right).
\end{align*}
\end{lem}

\begin{prop}
[Sorkin density function]\label{prop:SorkinDensity} Assume $\cA\cong\Mat(d,\C)$.
Let $\vec{x}=(x_{1},...,x_{N})\in\Manifold^{N}$, $\vec{X}=(x_{0},...,x_{N+1})\in\Manifold^{N+2}$,
and similarly for $\vec{y}$ and $\vec{Y}$, but let $y_{0}:=x_{0}$
and $y_{N+1}:=x_{N+1}$. We write briefly
\[
\int_{\vec{S}}\d\vec{x}:=\int_{S_{1}}\d\mu(x_{1})\cdots\int_{S_{N}}\d\mu(x_{N}).
\]
Then there exist two functions $\rho_{1}:\Manifold^{N+2}\to[0,\infty)$
and $\rho_{2}:\Manifold^{N+2}\times\Manifold^{N+2}\to[0,\infty)$,
called \termi{1-path and 2-path Sorkin density functions} respectively,
such that for all $S_{i}\in\Borel(\Manifold)$, $i=1,...,N$,
\begin{align}
\Prob(S_{1},...,S_{N}) & =\frac{1}{2d}\int_{\Manifold}\d\mu(x_{0})\int_{\Manifold}\d\mu(x_{N+1})\left[\int_{\vec{S}}\d\vec{x}\int_{\vec{S}}\d\vec{y}\,\rho_{2}(\vec{X},\vec{Y})-\mu(\vec{S})\int_{\vec{S}}\d\vec{x}\,\rho_{1}(\vec{X})\right],\label{eq:SorkinDensity1}
\end{align}
where $\mu(\vec{S}):=\prod_{i=1}^{N}\mu(S_{i})$, or equivalently
\begin{equation}
\Prob(S_{1},...,S_{N})=\frac{1}{2d}\int_{\Manifold}\d\mu(x_{0})\int_{\Manifold}\d\mu(x_{N+1})\int_{\vec{S}}\d\vec{x}\int_{\vec{S}}\d\vec{y}\left(\rho_{2}(\vec{X},\vec{Y})-\rho_{1}(\vec{X})-\rho_{1}(\vec{Y})\right).\label{eq:SorkinDensity}
\end{equation}
Note that $\Prob(S_{1},...,S_{N})=\Prob(\Manifold,S_{1},...,S_{N},\Manifold)$.
In fact, $\rho_{1}$ and $\rho_{2}$ are given as follows:
\begin{align*}
\SYM{\rho_{1}(\vec{X})}{rho1} & :=\Tr A(\vec{X})A(\vec{X})^{*},\\
\SYM{\rho_{2}(\vec{X},\vec{Y})}{rho2} & :=\Tr\left(A(\vec{X})+A(\vec{Y})\right)\left(A(\vec{X})+A(\vec{Y})\right)^{*},
\end{align*}
where $A(\vec{X}):=P_{x_{0}}\cdots P_{x_{N+1}}$.
\end{prop}

\begin{proof}
For each $x\in\Manifold$, let $v_{x}\in\cH\cong\C^{d}$ be such that
$P_{x}=|v_{x}\rangle\langle v_{x}|$. %
Let $F(\vec{X}):=\prod_{i=0}^{N}\langle v_{x_{i}}|v_{x_{i+1}}\rangle$.
{} Then we find
\[
\rho_{2}(\vec{X},\vec{Y})-\rho_{1}(\vec{X})-\rho_{1}(\vec{Y})=\Tr\left(A(\vec{X})A(\vec{Y})^{*}+A(\vec{Y})A(\vec{X})^{*}\right)=2\Re\left(F(\vec{X})\ol{F(\vec{Y})}\right),
\]
and hence
\begin{align*}
\Prob_{N}(S_{1}\times\cdots\times S_{N}) & =d^{-1}\Tr\Vac(S_{1})\cdots\Vac(S_{N})\Vac(S_{N})\cdots\Vac(S_{1})\\
 & =d^{-1}\int_{\Manifold}\d\mu(x_{0})\int_{\Manifold}\d\mu(x_{N+1})\ \Tr P_{x_{0}}\Vac(S_{1})\cdots\Vac(S_{N})P_{x_{N+1}}\Vac(S_{N})\cdots\Vac(S_{1})\\
 & =d^{-1}\int_{\Manifold}\d\mu(x_{0})\int_{\Manifold}\d\mu(x_{N+1})\int_{\vec{S}}\d\vec{x}\int_{\vec{S}}\d\vec{y}\ F(\vec{X})\ol{F(\vec{Y})}\\
 & =d^{-1}\int_{\Manifold}\d\mu(x_{0})\int_{\Manifold}\d\mu(x_{N+1})\int_{\vec{S}}\d\vec{x}\int_{\vec{S}}\d\vec{y}\ \Re\left(F(\vec{X})\ol{F(\vec{Y})}\right)\\
 & =d^{-1}\frac{1}{2}\int_{\Manifold}\d\mu(x_{0})\int_{\Manifold}\d\mu(x_{N+1})\int_{\vec{S}}\d\vec{x}\int_{\vec{S}}\d\vec{y}\ \left(\rho_{2}(\vec{X},\vec{Y})-\rho_{1}(\vec{X})-\rho_{1}(\vec{Y})\right).
\end{align*}
\end{proof}

\begin{rem}
The value of the complex function $F(\vec{X})$ defined in the above
proof is not observable or empirical, but the values of the Sorkin
density functions $\rho_{1}$ and $\rho_{2}$ are considered to be
``observable'' in a similar sense that the classical probability
density functions are ``observable''. Very roughly speaking, $\rho_{1}$
corresponds to the observed particle density in a single-slit experiment,
and $\rho_{2}$ corresponds to the observed particle density in a
double-slit experiment. The Sorkin additivity implies that the probabilistic
law in the ``$n$-slit experiment'' is fundamental for $n=1,2$
only, but non-fundamental (or derivative) for $n\ge3$; that is, all
empirical laws for the multi-slit experiments can be derived from
those of the 1-slit and 2-slit experiments.
\end{rem}

\section{Finite-dimensional CAR algebra}

Assume again that $\dim V=2n$, $n\in\N$. %
{} Each $T\in SO(V)$ induces the automorphism $\SYM{\sigma_{T}}{sigmaT}:\CAR(V)\to\CAR(V)$
such that $\sigma_{T}(\phi(v))=\phi(Tv)$.

We see that $\SYM{\fg_{V}}{gV}:=\Span_{\R}\{\phi_{u,v}|(u,v)\in\ON_{2}(V)\}$
gives a representation of the Lie algebra $\mathfrak{so}(V)\cong\mathfrak{so}(2n,\R)=\{X\in{\rm Mat}(2n,\R)|X^{\TT}=-X\}$.%
{} Let $\{e_{1},...,e_{2n}\}$ be an orthonormal basis of $(V,(\cdot|\cdot))$,
and
\[
\phi_{kl}:=\phi(e_{k})\phi(e_{l}),\qquad1\le k<l\le2n.
\]
Then $\{\tilde{\phi}_{kl}|1\le k<l\le2n\}$, $\tilde{\phi}_{kl}:=\phi_{kl}/2$
is a basis of $\fg_{V}$, satisfying
\[
[\tilde{\phi}_{ij},\tilde{\phi}_{kl}]=\delta_{jk}\tilde{\phi}_{il}+\delta_{il}\tilde{\phi}_{jk}+\delta_{lj}\tilde{\phi}_{ki}+\delta_{ki}\tilde{\phi}_{lj},
\]
which is the commutation relation for the usual basis of $\mathfrak{so}(2n,\R)$.
This determines the canonical representation $\d\pi_{V}:\mathfrak{so}(V)\to\fg_{V}$,
given by
\begin{equation}
\SYM{\d\pi_{V}}{dpiV}(X)=\frac{1}{4}\sum_{i,j=1}^{2n}X_{ij}\phi_{i}\phi_{j}=\frac{1}{2}\sum_{i,j=1}^{2n}X_{ij}\tilde{\phi}_{ij},\qquad X\in\mathfrak{so}(2n).\label{eq:pi(X)=00003D}
\end{equation}
This representation $\d\pi_{V}$ is independent of the orthonormal
basis $\{e_{1},...,e_{2n}\}$: %

For $u,v\in V$, define the operator $F_{u,v}$ on $V$ by $\SYM{F_{u,v}}{Fuv}:=|u)(v|$,
i.e., $F_{u,v}w:=(v|w)u$ for $w\in V$. If $(u|v)=0$, $\SYM{E_{u,v}}{Euv}:=F_{u,v}-F_{v,u}$
is an element of $\mathfrak{so}(V)$, and we see $\d\pi_{V}(E_{u,v})=\phi(u)\phi(v)/2$.

The group %
$G:=\exp(\fg_{V})=\{\exp(X)|X\in\fg_{V}\}\subset\CAR(V)$ is isomorphic
to $\Spin(V)\cong\Spin(2n,\R)$, the universal covering group of $SO(V)\cong SO(2n,\R)$,
where $SO(V)\cong\Spin(V)/\{\pm\bbOne\}$. We identify $G$ with $\SYM{\Spin(V)}{Spin()}$,
and write $\SYM{\spin(V)}{spin()}:=\fg_{V}$. For $U\in\Spin(V)$,
let $\SYM{\tilde{U}}{Util}:=\{\pm U\}$ denote the corresponding element
of $SO(V)$. Then we have
\[
\sigma_{\tilde{U}}(X)=U^{-1}XU,\qquad\text{for all }U\in\Spin(V).
\]

{}

Let $\SYM{\OCpl(V)}{OC}$ be the set of orthogonal complex structures
on $V$, i.e., $\ComplexS\in\OCpl(V)$ iff $\ComplexS\in SO(V)$ and
$\ComplexS^{2}=-\Id_{V}$. Let $\ComplexS\in\OCpl(V)$. We also see
$\ComplexS\in\mathfrak{so}(V)$ with $\ComplexS=e^{(\pi/2)\ComplexS}$.
The space $(V,\ComplexS)$ becomes a complex linear space by
\[
(x+\im y)v:=xv+\ComplexS yv,\qquad x,y\in\R,\ v\in V,
\]
and $(V,\ComplexS)$ is given a complex inner product
\[
\langle u|v\rangle\equiv\SYM{\langle u|v\rangle_{\ComplexS}}{<>J}:=(u|v)+\im(\ComplexS u|v),\qquad u,v\in V.
\]
Let $\SYM{U_{\ComplexS}(V)}{UJ()}$ be the group of unitary operators
on $(V,\langle\cdot|\cdot\rangle_{\ComplexS})$, i.e.,
\begin{equation}
U_{\ComplexS}(V):=\{g\in O(V)|\ComplexS g=g\ComplexS\}.\label{eq:def:UJ(V)}
\end{equation}
In fact, $O(V)$ may be replaced by $SO(V)$ in (\ref{eq:def:UJ(V)}),
since $U_{\ComplexS}(V)$ is a connected group.

{}

The $O(V)$ acts on $\OCpl(V)$ by $g\cdot\ComplexS:=g\ComplexS g^{-1}$,
$g\in O(V)$, $\ComplexS\in\OCpl(V)$. Then the stabilizer of $O(V)$
at $\ComplexS$ is $U_{\ComplexS}(V)$. Thus we see $\OCpl(V)\cong O(V)/U_{\ComplexS}(V)$
for each $\ComplexS$.%
{} Let $\SYM{\OCpl_{1}(V)}{OC1}\cong SO(V)/U_{\ComplexS}(V)$ be one
of two connected components of $\OCpl(V)$.

For any $\ComplexS\in\OCpl(V)$, the spectrum of $\im\cdot\d\pi_{V}(\ComplexS)\in\CAR(V)$
turns out to be
\begin{equation}
\spec(\im\cdot\d\pi_{V}(\ComplexS))=\left\{ k-\frac{n}{2}|k=0,...,n\right\} .\label{eq:spec(ipi(J))}
\end{equation}
Let $\SYM{\Vac_{\ComplexS}}{EJ0}$ be the spectral projection of $\im\cdot\d\pi_{V}(\ComplexS)$
w.r.t.~the eigenvalue $-\frac{n}{2}\in\spec(\im\cdot\d\pi_{V}(\ComplexS))$.
Define $N_{\ComplexS}\in\CAR(V)$ by $\SYM{N_{\ComplexS}}{NJ}:=\im\cdot\d\pi_{V}(\ComplexS)+\frac{n}{2}$,
called the \termi{number operator} w.r.t.~$\ComplexS\in\OCpl(V)$.
We have various expressions for $\Vac_{\ComplexS}$:
\begin{equation}
\Vac_{\ComplexS}=\lim_{\epsilon\to0}\epsilon\left(N_{\ComplexS}+\epsilon\right)^{-1}=\lim_{t\to\infty}e^{-tN_{\ComplexS}}=\frac{1}{n!}\left(1-N_{\ComplexS}\right)\cdots\left(n-N_{\ComplexS}\right).\label{eq:EJ0}
\end{equation}

\begin{prop}
\label{prop:UJ(V)}We have
\[
\{g\in SO(V)|\sigma_{g}(\Vac_{\ComplexS})=\Vac_{\ComplexS}\}=U_{\ComplexS}(V).
\]
\end{prop}

\begin{cor}
The map $\OCpl_{1}(V)\to\CAR(V)$, $\ComplexS\mapsto\Vac_{\ComplexS}$,
is injective.
\end{cor}

\begin{cor}
Let $\ComplexS\in\OCpl_{1}(V)$ and $\Manifold:=\{\sigma_{g}(\Vac_{\ComplexS})|g\in SO(V)\}$.
Then $\Manifold$ does not depend on $\ComplexS$, and $\Manifold\cong\OCpl_{1}(V)\cong SO(V)/U_{\ComplexS}(V)$
as smooth manifolds (in fact, as K\"ahler manifolds).
\end{cor}

\begin{rem}
It is often said that (quantum) fermionic systems lack classical counterparts.
This will be partially true, but we stress that the symplectic manifold
$\Manifold\cong\OCpl(V)\cong SO(V)/U_{\ComplexS}(V)$ can be interpreted
as the classical phase space of a fermionic system, at least at the
mathematical level. %
\end{rem}

Although Proposition \ref{prop:UJ(V)} is considered to be a standard
fact, it seems that its direct proof is rarely found in the literature.
(Here, the term ``direct proof'' refers to an elementary proof without
the general theory of unitary representations of compact Lie groups,
e.g., that of Cartan\textendash Weyl and/or that of Borel\textendash Weil.)%
{} %
We find that Proposition \ref{prop:UJ(V)} follows from Lemma \ref{lem:=00005BX,E=00005D=00003D0}
and Lemma \ref{lem:=00005BA,J=00005D=00003D0} below, which are directly
verified by elementary calculations.

Let
\[
A_{\ComplexS}(f):=\frac{1}{2}\left(\phi(f)+\im\phi(\ComplexS f)\right),\qquad A_{\ComplexS}^{*}(f):=\frac{1}{2}\left(\phi(f)-\im\phi(\ComplexS f)\right),\qquad f\in V,
\]
called \termi{annihilation} and \termi{creation operators} w.r.t.~$\ComplexS$,
respectively. Note that $f\mapsto A_{\ComplexS}(f)$ (resp.~$f\mapsto A_{\ComplexS}^{*}(f)$)
is antilinear (resp.~linear), i.e., $A_{\ComplexS}(\ComplexS f)=-\im A_{\ComplexS}(f)$
(resp.~$A_{\ComplexS}^{*}(\ComplexS f)=\im A_{\ComplexS}^{*}(f)$).%
{} We see
\[
\{A_{\ComplexS}(f),A_{\ComplexS}(g)\}=0,\qquad\{A_{\ComplexS}(f),A_{\ComplexS}(g)^{*}\}=\langle f|g\rangle_{\ComplexS}\bbOne.
\]
Let
\[
N_{\ComplexS,k}:=A_{\ComplexS}^{*}(e_{k})A_{\ComplexS}(e_{k})=\frac{1}{2}\left[\bbOne+\im\phi(e_{k})\phi(\ComplexS e_{k})\right]=\frac{1}{2}\left[\bbOne+\im\phi_{k}\phi_{k+n}\right],\qquad k=1,...,n,
\]
then we find
\[
N_{\ComplexS}=N_{\ComplexS,1}+\cdots+N_{\ComplexS,n}=\frac{1}{2}\left[n\bbOne+\im\sum_{k=1}^{n}\phi(e_{k})\phi(\ComplexS e_{k})\right]=\frac{1}{2}\left[n\bbOne+\im\sum_{k=1}^{n}\phi_{k}\phi_{k+n}\right],
\]
and hence the spectrum of $N_{\ComplexS}$ is $\{0,1,...,n\}$, and
(\ref{eq:spec(ipi(J))}). %

Any $X\in\spin(V)$ is expressed as follows.
\[
X=\sum_{1\le i<j\le2n}r_{ij}\phi_{i}\phi_{j},\qquad r_{ij}\in\R
\]

\begin{lem}
\label{lem:=00005BX,E=00005D=00003D0}Let $X=\sum_{1\le i<j\le2n}r_{ij}\phi_{i}\phi_{j}$.
Then, $[X,\Vac_{\ComplexS}]=0$ if and only if
\begin{equation}
r_{ij}=r_{i+n,j+n},\ r_{i,j+n}=r_{j,i+n},\qquad1\le i<j\le n.\label{eq:rij=00003Dri+nj+n}
\end{equation}
\end{lem}

\begin{proof}
(outline) By a straightforward calculation, we have
\[
[X,\Vac_{\ComplexS}]=\sum_{1\le i<j\le2n}r_{ij}\left[A_{\ComplexS}^{*}(e_{i})A_{\ComplexS}^{*}(e_{j})\Vac_{\ComplexS}+\Vac_{\ComplexS}A_{\ComplexS}(e_{j})A_{\ComplexS}(e_{i})\right].
\]
This implies that
\[
[X,\Vac_{\ComplexS}]=0\quad\text{iff}\quad\sum_{1\le i<j\le2n}r_{ij}A_{\ComplexS}^{*}(e_{i})A_{\ComplexS}^{*}(e_{j})=0.
\]
Moreover we have
\[
\sum_{1\le i<j\le2n}r_{ij}A_{\ComplexS}^{*}(e_{i})A_{\ComplexS}^{*}(e_{j})=\sum_{1\le i<j\le n}\left(r_{ij}-r_{i+n,j+n}+\im r_{i,j+n}-\im r_{j,i+n}\right)A_{\ComplexS}^{*}(e_{i})A_{\ComplexS}^{*}(e_{j}).
\]
Thus (\ref{eq:rij=00003Dri+nj+n}) follows.
\end{proof}
\begin{lem}
\label{lem:=00005BA,J=00005D=00003D0}Let $E_{ij}$ ($1\le i<j\le2n$)
be the standard basis of $\mathfrak{so}(2n)$, that is, $E_{ij}:=F_{ij}-F_{ji}$,
$F_{ij}:=|e_{i})(e_{j}|$. Let
\[
A=\sum_{1\le i<j\le2n}r_{ij}E_{ij}\in\mathfrak{so}(V),\qquad\ComplexS=\sum_{i=1}^{n}E_{i,i+n}.
\]
Then we have%
\begin{equation}
[A,\ComplexS]=0\quad\text{iff}\quad r_{i,j+n}=r_{j,i+n},\ r_{ij}=r_{i+n,j+n}\qquad1\le i<j\le n.\label{eq:=00005BA,J=00005D=00003D0iff}
\end{equation}
\end{lem}

\begin{proof}
By a straightforward calculation, we have
\begin{equation}
[A,\ComplexS]=\sum_{1\le i<j\le n}\left(\left(r_{i,j+n}-r_{j,i+n}\right)\left(-E_{ij}+E_{i+n,j+n}\right)+\left(r_{ij}-r_{i+n,j+n}\right)\left(E_{i,j+n}-E_{j,i+n}\right)\right).\label{eq:=00005BA,J=00005D}
\end{equation}
Hence (\ref{eq:=00005BA,J=00005D=00003D0iff}) follows.
\end{proof}
For each $\ComplexS\in\OCpl(V)$, the tangent space $T_{\ComplexS}\OCpl(V)$
at $\ComplexS$ is naturally identified with the subspace $[\ComplexS,\mathfrak{so}(2n,\R)]$
of $\mathfrak{so}(2n,\R)$. If we choose the orthonormal basis $\{e_{1},...,e_{2n}\}$
where $\ComplexS=\sum_{i=1}^{n}E_{i,i+n}$, we find from (\ref{eq:=00005BA,J=00005D})
that
\[
[\ComplexS,\mathfrak{so}(2n)]=\Span_{\R}\left\{ X_{ij},Y_{ij}|\ 1\le i<j\le n\right\} ,\qquad\SYM{X_{ij}}{Xij}:=-E_{ij}+E_{i+n,j+n},\ \SYM{Y_{ij}}{Yij}:=E_{i,j+n}-E_{j,i+n}.
\]
The equation $\ComplexS^{2}=-1$ implies $\ad(\ComplexS)^{3}=-4\ad(\ComplexS)$,
that is,
\begin{equation}
\ad(\ComplexS)^{2}([\ComplexS,X])=[\ComplexS,[\ComplexS,[\ComplexS,X]]]=-4[\ComplexS,X],\qquad\text{for all }X\in\mathfrak{so}(2n,\R).\label{eq:=00005BJ=00005BJ=00005BJ}
\end{equation}
Let
\[
\SYM{\scJ}{Jsc}:=\frac{1}{2}\ad(\ComplexS)\Big|_{[\ComplexS,\mathfrak{so}(2n)]},
\]
then (\ref{eq:=00005BJ=00005BJ=00005BJ}) means $\scJ^{2}=-1$, and
hence $\scJ$ determines the canonical complex structure on $T_{\ComplexS}\OCpl(V)=[\ComplexS,\mathfrak{so}(2n)]$.
Explicitly, we find
\[
\scJ X_{ij}=Y_{ij},\qquad\scJ Y_{ij}=-X_{ij},
\]
from
\[
[\ComplexS,E_{ij}]=-[\ComplexS,E_{i+n,j+n}]=-Y_{ij},\qquad[\ComplexS,E_{i,j+n}]=-X_{ij},\qquad i,j=1,...,n.
\]
{}

\section{Empirical laws in the $n=2$ case}

\label{subsec:4-dim}

The explicit $2^{n}\times2^{n}$ matrix expression of $\Vac_{\ComplexS}$
is almost impossible for large $n$. However, in the simplest nontrivial
case where $\dim V=2n=4$, the explicit matrix expression is not difficult.

Let $V=\R^{4}$. We often use the following two ``standard'' complex
structures on $\R^{4}$:

\[
\ComplexS_{1}:=\begin{pmatrix}0 & -1 & 0 & 0\\
1 & 0 & 0 & 0\\
0 & 0 & 0 & -1\\
0 & 0 & 1 & 0
\end{pmatrix},\qquad\ComplexS_{2}:=\begin{pmatrix}0 & 0 & -1 & 0\\
0 & 0 & 0 & -1\\
1 & 0 & 0 & 0\\
0 & 1 & 0 & 0
\end{pmatrix}
\]
These belong to different connected components of $\OCpl(\R^{4})$.
Generally, any element of two components of $\OCpl(\R^{4})$ is expressed
respectively by%
\[
\ComplexS_{1,x}:=\begin{pmatrix}0 & x_{3} & -x_{2} & x_{1}\\
-x_{3} & 0 & x_{1} & x_{2}\\
x_{2} & -x_{1} & 0 & x_{3}\\
-x_{1} & -x_{2} & -x_{3} & 0
\end{pmatrix},\qquad\ComplexS_{2,x}:=\begin{pmatrix}0 & x_{3} & -x_{2} & -x_{1}\\
-x_{3} & 0 & x_{1} & -x_{2}\\
x_{2} & -x_{1} & 0 & -x_{3}\\
x_{1} & x_{2} & x_{3} & 0
\end{pmatrix},\qquad x\in S^{2},
\]
where $S^{2}:=\{x=(x_{1},x_{2},x_{3})\in\R^{3}|\sum_{k=1}^{3}x_{k}^{2}=1\}$.
Notice that $[\ComplexS_{1,x},\ComplexS_{2,y}]=0$ and $\Tr\ComplexS_{1,x}\ComplexS_{2,y}=0$
for all $x,y\in S^{2}$. Let $\sigma_{k}$, $k=1,2,3$ be the standard
Pauli matrices, with $\sigma_{0}=\begin{pmatrix}1 & 0\\
0 & 1
\end{pmatrix}$. Let $e_{i}$, $i=1,...,4$ be the standard basis of $\R^{4}$, and
consider the representation of $\phi_{i}=\phi(e_{i})$ given by $\phi_{1}=\sigma_{1}\otimes\sigma_{0},$
$\phi_{2}=\sigma_{2}\otimes\sigma_{0},\ \phi_{3}=\sigma_{3}\otimes\sigma_{1},\ \phi_{4}=\sigma_{3}\otimes\sigma_{2},$
more explicitly
\[
\phi_{1}=\begin{pmatrix}0 & \sigma_{0}\\
\sigma_{0} & 0
\end{pmatrix},\quad\phi_{2}=\begin{pmatrix}0 & -\im\sigma_{0}\\
\im\sigma_{0} & 0
\end{pmatrix},\quad\phi_{3}=\begin{pmatrix}\sigma_{1} & 0\\
0 & -\sigma_{1}
\end{pmatrix},\quad\phi_{4}=\begin{pmatrix}\sigma_{2} & 0\\
0 & -\sigma_{2}
\end{pmatrix}.
\]
We see
\[
\phi_{5}:=\phi_{1}\phi_{2}\phi_{3}\phi_{4}={\rm diag}(-1,1,1,-1)
\]
commutes with $\phi_{kl}:=\phi(e_{k})\phi(e_{l})$ for all $1\le k<l\le4$;
The canonical representation of $\CAR(\R^{4})_{{\rm even}}$ on $\C^{4}$
is decomposed to the irreducible ones on the eigenspaces of $\phi_{5}$.
By (\ref{eq:pi(X)=00003D}) we have%
\[
\d\pi(\ComplexS_{1,x})=\begin{pmatrix}\im x_{3} & 0 & 0 & x_{2}+\im x_{1}\\
0 & 0 & 0 & 0\\
0 & 0 & 0 & 0\\
-x_{2}+\im x_{1} & 0 & 0 & -\im x_{3}
\end{pmatrix},\qquad\d\pi(\ComplexS_{2,x})=\begin{pmatrix}0 & 0 & 0 & 0\\
0 & \im x_{3} & x_{2}+\im x_{1} & 0\\
0 & -x_{2}+\im x_{1} & -\im x_{3} & 0\\
0 & 0 & 0 & 0
\end{pmatrix},
\]
and by (\ref{eq:EJ0}),%
\[
\Vac_{\ComplexS_{1,x}}=\frac{1}{2}\begin{pmatrix}x_{3}+1 & 0 & 0 & x_{1}-\im x_{2}\\
0 & 0 & 0 & 0\\
0 & 0 & 0 & 0\\
x_{1}+\im x_{2} & 0 & 0 & -x_{3}+1
\end{pmatrix},\qquad\Vac_{\ComplexS_{2,x}}=\frac{1}{2}\begin{pmatrix}0 & 0 & 0 & 0\\
0 & x_{3}+1 & x_{1}-\im x_{2} & 0\\
0 & x_{1}+\im x_{2} & -x_{3}+1 & 0\\
0 & 0 & 0 & 0
\end{pmatrix}.
\]

Define the prior conditional probability $\Prob_{i}(y|x)$ ($i=1,2$)
as follows.
\[
\Prob_{i}(y|x)\equiv\Prob(\Vac_{\ComplexS_{i,y}}|\Vac_{\ComplexS_{i,x}}):=\frac{\Tr\Vac_{\ComplexS_{i,x}}\Vac_{\ComplexS_{i,y}}}{\Tr\Vac_{\ComplexS_{i,x}}}=\Tr\Vac_{\ComplexS_{i,x}}\Vac_{\ComplexS_{i,y}},\qquad x,y\in S^{2}.
\]
From the above explicit matrix expressions, we have%
\[
\Prob_{i}(y|x)=\frac{1}{2}\left((x|y)+1\right)=-\frac{1}{8}\Tr\ComplexS_{i,x}\ComplexS_{i,y}+\frac{1}{2},\qquad x,y\in S^{2},\ i=1,2,
\]
which may be said to be an empirical law.

Consider $\Vac_{\ComplexS_{i,x}}$ for $i=1$ only, so that we can
redefine $\Vac_{\ComplexS_{1,x}}$ to be the $2\times2$ matrix%
\[
\Vac_{x}\equiv\Vac_{\ComplexS_{1,x}}:=\frac{1}{2}\begin{pmatrix}1+x_{3} & x_{1}-\im x_{2}\\
x_{1}+\im x_{2} & 1-x_{3}
\end{pmatrix},
\]
acting on the eigenspace of $\phi_{5}$ w.r.t.~the eigenvalue $-1$.
This is nothing but the spin $\frac{1}{2}$ state with spin direction
$x\in S^{2}$:
\[
\Vac_{x}=\frac{1}{2}\left(x_{1}\sigma_{1}+x_{2}\sigma_{2}+x_{3}\sigma_{3}+1\right).
\]
Thus the statistics of the set of projections $\{\Vac_{\ComplexS_{1,x}}|x\in S^{2}\}$
is well known.

Let $\Manifold:=S^{2}\cong\OCpl_{1}(\R^{4})$. Let $\mu$ be a invariant
measure on $\Manifold$. For $S\in\Borel(\Manifold)$, let $\Vac(S)\equiv\Vac_{S}:=\int_{S}\Vac_{x}\,\d\mu(x)$.
Here we assume that $\mu$ is normalized so that $\Vac(\Manifold)=\bbOne$.

Let $\vec{X}=(x_{0},...,x_{N+1}),\vec{Y}=(y_{0},...,y_{N+1})\in\Manifold^{N+2}$,
but assume $y_{0}=x_{0}$, $y_{N+1}=x_{N+1}$. Let
\[
\SYM{\tau(x_{1},x_{2})}{tau()}:=\Prob(x_{2}|x_{1})=\Tr\Vac_{x_{1}}\Vac_{x_{2}}=\frac{(x_{1}|x_{2})+1}{2}=1-\frac{1}{4}\|x_{1}-x_{2}\|^{2}.
\]
Then the 1-path Sorkin density function $\rho_{1}$ defined in Proposition
\ref{prop:SorkinDensity} is given simply by
\begin{equation}
\rho_{1}(\vec{X})=\Tr A(\vec{X})A(\vec{X})^{*}=\prod_{i=0}^{N}\tau(x_{i},x_{i+1}),\qquad A(\vec{X}):=\Vac_{x_{0}}\cdots\Vac_{x_{N+1}}..\label{eq:rho1-4D}
\end{equation}
The 2-path Sorkin density function $\rho_{2}$ is more nontrivial.
For $\vec{x}=(x_{1},...,x_{N})\in\Manifold^{N}$, it turns out that
\begin{equation}
\Tr A(\vec{x})=R(\vec{x})e^{\im\area(\vec{x})/2},\qquad R(\vec{x}):=\prod_{i=1}^{N-1}\sqrt{\tau(x_{i},x_{i+1})},\label{eq:TrA(x)=00003D}
\end{equation}
where, roughly speaking, $\area(\vec{x})$ is the signed area of the
sphere polygon, enclosed by the geodesic lines between $x_{i},x_{i+1}$
($i=1,...,N$) with $x_{N+1}:=x_{1}$. Here the area is measured by
the usual measure such that the area of whole sphere $S^{2}$ is $4\pi$,
not by $\mu$. Note that the area $\area(\vec{x})$ of a sphere polygon
has an indeterminacy; there are two regions enclosed by a simple closed
curve on a 2-sphere. However, if the signed area of one of two regions
enclosed by an oriented closed curve is $a\in\R$, the signed area
of the other region is $a-4\pi$. Hence the signed area becomes well-defined
when we consider its value to be not in $\R$ but in $\R/4\pi$. In
this case, $e^{\im\area(\vec{x})/2}\in\C$ is also well-defined.

When $N=3$, (\ref{eq:TrA(x)=00003D}) is shown by elementary calculations.%
{} Let $\vec{x}=(x_{1},x_{2},x_{3})$, and $\area(\vec{x})$ denote
the area of the spherical triangle $x_{1}x_{2}x_{3}$. A theorem in
spherical geometry says that
\[
\cos\frac{\area(\vec{x})}{2}=\frac{1+(x_{1}|x_{2})+(x_{2}|x_{3})+(x_{3}|x_{1})}{\sqrt{2\left(1+(x_{1}|x_{2})\right)\left(1+(x_{2}|x_{3})\right)\left(1+(x_{3}|x_{1})\right)}},
\]
if $(x_{i}|x_{j})\neq-1$, $i,j=1,2,3$. On the other hand, we can
check
\[
\Tr A(\vec{x})=\frac{1}{4}\left(1+(x_{1}|x_{2})+(x_{2}|x_{3})+(x_{3}|x_{1})+\im\cdot\det\left[x_{1}x_{2}x_{3}\right]\right),
\]
\[
\left|\Tr A(\vec{x})\right|^{2}=\frac{1}{8}\left(1+(x_{1}|x_{2})\right)\left(1+(x_{2}|x_{3})\right)\left(1+(x_{3}|x_{1})\right).
\]
Thus we have%
{} $\Re\,\Tr A(\vec{x})=\left|\Tr A(\vec{x})\right|\Re e^{\im\area(\vec{x})/2}$,
and hence
\[
\Tr A(\vec{x})=\left|\Tr A(\vec{x})\right|e^{\pm\im\area(\vec{x})/2}.
\]
Here, the indefinite sign $\pm$ can be fixed to $+$ if we suitably
determine the orientation of the triangle.

However, the above calculation is specific to the 2-sphere $\Manifold\cong S^{2}$,
and so it will not be generalized for $\Manifold\cong\OCpl_{1}(V)$
with a higher dimension $\dim V=2n\ge6$. Even if we confine ourselves
to the 2-sphere, for general $N\in\N$, the definition of $\area(\vec{x})$
is not clear yet, since the closed piecewise geodesic curve $(x_{1},...,x_{N},x_{1})$
may not be a \emph{simple} closed curve. Actually, it is easier to
give a rigorous definition of $U(\vec{x}):=e^{\im\area(\vec{x})/2}$
than that of $\area(\vec{x})$ itself. A natural definition of $U(\vec{x})$
is as the holonomy of a closed curve on $\Manifold$, with respect
to a $U(1)$-principal bundle over $\Manifold$ (or the associated
complex line bundle) with a connection. Such line bundle is called
the \termi{prequantization bundle} of $(\Manifold,\omega)$, where
$\omega$ is a symplectic 2-form on $\Manifold$ (If $\Manifold=S^{2}$,
$\omega$ is the usual volume 2-form on $S^{2}$ up to scalar factor).
If $\Manifold$ is simply connected, the prequantization bundle of
$\Manifold$ is uniquely determined, if it exists.\cite{Woo92}%
{} However in our case, the system is ``quantized'' from the outset;
Each $\Vac_{x}$ is already represented as an operator on the Hilbert
space $\cK:=\C^{2}$, which is interpreted as the quantum state space.
Identify $\Manifold$ with $\{\Vac_{x}|x\in S^{2}\}$, then the fiber
of the $U(1)$-principal bundle at $\Vac_{x}\in\Manifold$ is given
by $\{v\in\cK|\Vac_{x}v=v,\ \|v\|=1\}$.

We have not yet finished the definition of $U(\vec{x})$ (or $\area(\vec{x})$),
because there are at least two geodesic lines between two points $x,y\in S^{2}$;
Exactly speaking, the term ``piecewise geodesic curve $(x_{1},...,x_{N},x_{1})$''
is not uniquely defined. %

Let us temporarily simply assume that $e^{\im\area(\vec{Z})/2}$ is
well-defined. Let
\[
\vec{Z}\equiv(z_{1},...,z_{2N+2}):=(x_{0},...,x_{N+1},y_{N},...,y_{1}),
\]
then we have
\[
\Tr\left(A(\vec{X})A(\vec{Y})^{*}\right)=\Tr A(\vec{Z})=R(\vec{Z})e^{\im\area(\vec{Z})/2},
\]
and hence
\[
\rho_{2}(\vec{X},\vec{Y})-\rho_{1}(\vec{X})-\rho_{1}(\vec{Y})=2\Re\,\Tr\left(A(\vec{X})A(\vec{Y})^{*}\right)=2R(\vec{Z})\Re e^{\im\area(\vec{Z})/2}.
\]
Therefore the 2-path Sorkin density is given by
\begin{equation}
\rho_{2}(\vec{X},\vec{Y})=2R(\vec{Z})\cos(\area(\vec{Z})/2)+\rho_{1}(\vec{X})+\rho_{1}(\vec{Y}).\label{eq:rho2-4D}
\end{equation}

Eqs.~(\ref{eq:rho1-4D}) and (\ref{eq:rho2-4D}) themselves can be
seen as empirical laws, and furthermore, other empirical law are derived
from (\ref{eq:rho1-4D}) and (\ref{eq:rho2-4D}) by Proposition \ref{prop:SorkinDensity}.

Next, consider a piecewise smooth curve $\curve:[0,1]\to\Manifold$.
Let
\[
\SYM{A(\curve,L)}{A(ga,L)}:=\Vac_{\curve(0/L)}\Vac_{\curve(1/L)}\cdots\Vac_{\curve(L/L)},\qquad L\in\N,
\]
and
\begin{equation}
\SYM{A(\curve)}{A(gam)}:=\lim_{L\to\infty}A(\curve,L).\label{eq:def:A(C)}
\end{equation}
{} %
We see $A(\curve)A(\curve)^{*}=\Vac_{\curve(0)}$ and $A(\curve)^{*}A(\curve)=\Vac_{\curve(1)}$.
If $\curve_{1},\curve_{2}:[0,1]\to\Manifold$ are two piecewise smooth
curves such that $\curve_{1}(0)=\curve_{2}(0)$, $\curve_{1}(1)=\curve_{2}(1)$.
Then $A(\curve_{1})=e^{\im t}A(\curve_{2})$ for some $t\in\R$. %
{} %
When $\Vac_{\curve(0)}\Vac_{\curve(1)}\neq0$ (i.e., $\curve(0)+\curve(1)\neq0$),
let
\[
T_{\curve}:=\frac{\Vac_{\curve(0)}\Vac_{\curve(1)}}{\sqrt{\Tr\Vac_{\curve(0)}\Vac_{\curve(1)}}},
\]
then we see $A(\curve,L)=zT_{\curve}$ for some $z\in\C,\ |z|\le1$,
and $A(\curve)=zT_{\curve}$ for some $z\in U(1)$, i.e., $z\in\C,\ |z|=1.$

Although $A(\curve)$ might correspond to an ideal, non-realistic
measurement process, it is approximated by more realistic measurement
processes $A(\curve,L)$, $L\in\N$.

Let $\curve,\curve_{1}$ and $\curve_{2}$ be piecewise smooth curves
on $\Manifold$, satisfying $\curve_{1}(0)=\curve_{2}(0)$, $\curve_{1}(1)=\curve_{2}(1)$.
As in Proposition \ref{prop:SorkinDensity}, the 1-path and 2-path
Sorkin density functions $\rho_{1}(\curve)$ and $\rho_{2}(\curve_{1},\curve_{2})$
are naturally defined by
\begin{align*}
\SYM{\rho_{1}(\curve)}{rho1} & :=\Tr A(\curve)A(\curve)^{*},\\
\SYM{\rho_{2}(\curve_{1},\curve_{2})}{rho2} & :=\Tr\left(A(\curve_{1})+A(\curve_{2})\right)\left(A(\curve_{1})+A(\curve_{2})\right)^{*},
\end{align*}
but indeed we have simply $\rho_{1}(\curve)\equiv1$, which contains
very little information as an empirical law.

When $\curve(0)=\curve(1)$, define $\SYM{\hol(\curve)}{hol(C)}\in U(1)$
to be the holonomy of $\curve$, so that ``the signed area enclosed
by $\curve$'' $\area(\curve)\in\R/4\pi$ is determined by $\hol(\curve)=e^{\im\area(\curve)/2}$.
Let $\curve^{-1}(t):=\curve(1-t)$, and $\curve_{1}\bullet\curve_{2}^{-1}$
denote the concatenation of $\curve_{1}$ and $\curve_{2}^{-1}$,
which is a closed curve. We find that
\[
\Tr\left[A(\curve_{1})A(\curve_{2})^{*}\right]=\hol(\curve_{1}\bullet\curve_{2}^{-1}),
\]
and hence
\begin{equation}
\rho_{2}(\curve_{1},\curve_{2})=2+2\Re\,\hol(\curve_{1}\bullet\curve_{2}^{-1}).\label{eq:rho2-4D-hol-4D}
\end{equation}
Recall that the r.h.s.~of the ``empirical law'' (\ref{eq:rho2-4D})
is not rigourously defined yet. On the other hand, there is no ambiguity
in (\ref{eq:rho2-4D-hol-4D}), and moreover we may expect that (\ref{eq:rho2-4D-hol-4D})
can be generalized for higher dimensions $\dim V=2n\ge6$. However,
in Proposition \ref{prop:SorkinDensity}, $\Prob(S_{1},...,S_{N})$
is expressed by $\rho_{1}(\vec{X})$ and $\rho_{2}(\vec{X},\vec{Y})$
for fixed finite $N$.
\begin{problem}
\label{prob:rho2(ga,ga)}Can we express $\Prob(S_{1},...,S_{N})$
by $\rho_{2}(\curve_{1},\curve_{2})$ %
in (\ref{eq:rho2-4D-hol-4D}), for any $N\in\N$ and $S_{1},...,S_{N}\in\Borel(\Manifold)$?
\end{problem}

If the answer to the above problem is positive, one might be permitted
to say that (\ref{eq:rho2-4D-hol-4D}) is the \emph{most fundamental}
empirical law, because we can derive the other empirical laws concerning
$\Prob(S_{1},...,S_{N})$ from the single law (\ref{eq:rho2-4D-hol-4D})
(with the almost trivial law $\rho_{1}(\curve)\equiv1$). Actually
the answer will be positive, but the rigorous formulation of such
expression will be rather difficult, as follows.

A very simple rough idea is that $\Prob(S_{1},...,S_{N})$ is expressed
by a ``path integral'': Consider the following formal path integral
expression, analogous to Proposition \ref{prop:SorkinDensity}:
\begin{align}
\Prob(S_{1},...,S_{N}) & =\frac{1}{2}\int_{\Paths(\vec{S})}\cD\curve_{1}\int_{\Paths(\vec{S})}\cD\curve_{2}\left(\rho_{2}(\curve_{1},\curve_{2})-\rho_{1}(\curve_{1})-\rho_{1}(\curve_{2})\right)\\
 & =\int_{\Paths(\vec{S})}\cD\curve_{1}\int_{\Paths(\vec{S})}\cD\curve_{2}\ \Re\,\hol(\curve_{1}\bullet\curve_{2}^{-1}),\label{eq:P(SS)=00003DpathInt}
\end{align}
where we assume $\curve_{1}(0)=\curve_{2}(0)$, $\curve_{1}(1)=\curve_{2}(1)$,
and
\[
\Paths(\vec{S}):=\{\curve\in C([0,1],\Manifold)|\,\curve(j\epsilon)\in S_{j},\ j=1,...,N\},\qquad\epsilon:=\frac{1}{N+1}.
\]
Generally, rigorous justifications of path integrals used in quantum
physics are difficult. However it is well-known that some rigorous
formulations of path integrals are possible in terms of \emph{Brownian
motions}, mainly with the Feynman\textendash Kac formula and its generalizations.
Recall that we can define the Brownian motion on a manifold $\Manifold$
only if $\Manifold$ is furnished with a Riemannian metric. Since
a classical-mechanical phase space is represented as a symplectic
manifold, we cannot define the Brownian motion on a phase space in
general. However sometimes a phase space has a canonical K{\"a}hler
structure, which contains a Riemannian structure. (Recall that $\OCpl_{1}(V)\cong SO(2n)/U(n)$
is a K{\"a}hler manifold.)

For the rigorous considerations for the path integrals on phase spaces,
see Daubechies and Klauder \cite{DK85}, Charles \cite{Cha1999},
and Yamashita \cite{Yam11,Yam18,Yam22a,Yam22b}. See G{\"u}neysu
\cite{Gun10,Gun17} for the generalized Feynman\textendash Kac formula
on a vector bundle over a Riemannian manifold.

We remark that the above path integral quantization on the K{\"a}hler
manifold $\Manifold$ is intended to give the same physical/empirical
result as the geometric quantization with the complex polarization
on $\Manifold$, which is already well-established (see \cite{Woo92,Char2016a,LeF2018}
and references therein). Thus alternatively we could view the path
integral expression (\ref{eq:P(SS)=00003DpathInt}) as a formal calculus
of geometric quantization, without any measure-theoretical ground
such as the theory of Brownian motions.

\section{Empirical laws in general dimension}

\label{sec:Empirical-laws-ndim}

If the word ``express'' in Problem \ref{prob:rho2(ga,ga)} is understood
to include some sort of path integral expression, justified by the
generalized Feynman\textendash Kac formula, the answer to the problem
will be positive. Furthermore, the positive answer is expected to
be generalized to general dimension $\dim V=2n$.

Here we can avoid to use the tools from the theory of geometric quantization,
such as prequantization bundles \cite{Woo92}, since the system is
quantized from the outset; This means that the CAR algebra $\CAR(V)\cong\Mat(2^{n},\C)$
is canonically represented on the Hilbert space $\cH\cong\C^{2^{n}}$,
uniquely up to unitary equivalence. Let $\SYM{\bP(\cH)}{P(H)}$ be
the set of projections of rank 1 on $\cH$. Then $\bP(\cH)$ is identified
with the projective space $\cH/\C^{\times}\cong\C^{2^{n}}/\C^{\times}=P^{2^{n}-1}\C$,
which is a compact K{\"a}hler manifold. The tangent space $T_{P}\bP(\cH)$
of $\bP(\cH)$ at $P\in\bP(\cH)$ is canonically identified with the
linear subspace $\SYM{\cT_{P}}{TP}:=\{|u\rangle\langle v|:v\in\ran(P),\ u\in\ker(P)\}$
of $\SYM{\Bdd(\cH)}{Bd(H)}$ (the (bounded) operators on $\cH$),
which is given an inner product $\langle X|Y\rangle=\Tr X^{*}Y$,
$X,Y\in\cT_{P}$, together with the Riemannian metric (real inner
product) $(\cdot|\cdot):=\Re\langle\cdot|\cdot\rangle$ and the symplectic
form $\omega(\cdot,\cdot):=\Im\langle\cdot|\cdot\rangle$. (Precisely,
if $P,P'\in\bP(\cH)$ and $P\neq P'$, the intersection of two tangent
spaces $T_{P}\bP(\cH)\cap T_{P'}\bP(\cH)$ should be empty, whereas
$\cT_{P}\cap\cT_{P'}=\{0\}$. Thus the above $|u\rangle\langle v|$
should be replaced by $(P,|u\rangle\langle v|)$.)

Let $\SYM{\bbS(\cH)}{S(H)}:=\{v\in\cH:\|v\|=1\}$. Let $\curve:[0,1]\to\bP(\cH)$
(resp.~$c:[0,1]\to\bbS(\cH)$) be a piecewise smooth curve on $\bP(\cH)$
(resp.~$\bbS(\cH)$).%
{} The curve $c$ is called a horizontal lift of $\curve$ (or a parallel
transport along $\curve$) if $c(t)\in\ran(\curve(t))$ and $\langle c(t)|\dot{c}(t)\rangle=0$
for $t\in[0,1]$. If $\curve$ is closed, i.e., $\curve(0)=\curve(1)$,
and $c$ is a horizontal lift of $\curve$, $\SYM{\hol(\curve)}{hol}:=\langle c(0)|c(1)\rangle\in U(1)$
is called the holonomy along $\curve$.

For simplicity, let $\Manifold:=\{\Vac_{\ComplexS}|\ComplexS\in\OCpl_{1}(V)\}$,
instead of $\Manifold:=\OCpl_{1}(V)$. Then $\Manifold$ becomes a
K{\"a}hler submanifold of $\bP(\cH)$. For a piecewise smooth closed
curve $\curve:[0,1]\to\Manifold$, $\hol(\curve)$ is defined to be
the holonomy along the closed curve on $\bP(\cH)$.

Let $\mu$ be a $SO(2n)$-invariant measure on $\OCpl_{1}(V)$ (and
$\Manifold$), normalized so that
\[
\SYM{\Vac_{\Manifold}}{EM}:=\int_{\OCpl_{1}(V)}\Vac_{\ComplexS}\,\d\mu(\ComplexS)\,\left(=\int_{\Manifold}P\,\d\mu(P)\right)
\]
is a projection. Let $\cK:=\ran(\Vac_{\Manifold})$.

Let $\curve,\curve_{1},\curve_{2}:[0,1]\to\Manifold$ be piecewise
smooth curves on $\Manifold$, such that $\curve_{1}(0)=\curve_{2}(0)$,
$\curve_{1}(1)=\curve_{2}(1)$. Then the empirical laws will be derived
from the fundamental ones, i.e., from the Sorkin density functions
$\rho_{1}(\curve)\equiv1$ and
\begin{equation}
\rho_{2}(\curve_{1},\curve_{2})=2+2\Re\,\hol(\curve_{1}\bullet\curve_{2}^{-1}),\label{eq:rho2-hol}
\end{equation}
similarly to (\ref{eq:rho2-4D-hol-4D}) in the $\dim V=4$ case. More
explicitly, we will have the path integral representation of the prior
probability
\begin{align}
\Prob(S_{1},...,S_{N}) & =\int_{\Paths(\vec{S})}\cD\curve_{1}\int_{\Paths(\vec{S})}\cD\curve_{2}\ \Re\,\hol(\curve_{1}\bullet\curve_{2}^{-1}),\qquad S_{i}\in\Borel(\Manifold),\ i=1,...,N,\label{eq:pathInt-hol}
\end{align}
similarly to (\ref{eq:P(SS)=00003DpathInt}). However, we shall not
go into further detail regarding the rigorous foundation of this path
integral in this paper. We refer again to \cite{DK85,Cha1999,Yam11,Yam18,Yam22a,Yam22b}
for further information.

\subsection{Empirical law for the Majorana field}

Let $V:=\Gamma_{\cpt}^{\infty}(D_{\re}M)$ given the real pre-inner
product%
{} $(\cdot|\cdot)_{m}$, and consider $\cA_{{\rm Majorana}}=\CAR(V)$.

From a purely mathematical viewpoint, the observables $X_{f_{1},f_{2}}:=\im\Psi(f_{1})\Psi(f_{2})$
($(f_{1},f_{2})_{m}\in\ON_{2}(V)$) (recall (\ref{eq:def:Xuv})) will
be the most ``basic'' ones, since the subalgebra of observables
$\cA_{{\rm Majorana}}^{{\rm even}}\subset\cA_{{\rm Majorana}}$ is
generated by $X_{f_{1},f_{2}}$'s. Note that $\ON_{2}(V)$ was defined
by (\ref{eq:def:ON2}) when $(\cdot|\cdot)$ is an inner product,
but it is defined also when $(\cdot|\cdot)$ is a pre-inner product.
However, not all of $X_{f_{1},f_{2}}$'s are very ``familiar'' to
us, in the sense that we mentioned in Section \ref{sec:Introduction}.
For example, assume that $\supp(f_{1})$ and $\supp(f_{2})$ are spacelike
separated, say, that they are 1000\,km apart. Then it is not easy
for us to conceive the measurement procedure for $X_{f_{1},f_{2}}$.
(Recall that the field operator $\Psi(f_{i})$ is not observable for
each $i=1,2$, and hence we cannot reduce the measurement of $X_{f_{1},f_{2}}$
to each measurement of $\Psi(f_{i})$.)

By definition, an empirical law should be such that we can verify/falsify
it experimentally. Hence we hope that we can easily conceive the concrete
measurement procedure of it. However, as to the Majorana field (and
more general fermion fields), this hope does not seem to be fulfilled
immediately. Although it will be an important task to construct the
concrete/realistic measurement procedures, we will be satisfied with
a formulation at a more abstract level here; We will refer to $X_{f_{1},f_{2}}$
as an observable, for an arbitrary $(f_{1},f_{2})_{m}\in\ON_{2}(V)$,
hypothesizing the existence of the concrete measurement procedure
of it.

Next, we should mention the ``redundancy'' of the test function
space $V$. Let $\cN_{m}:=\{v\in V;\,(v|v)_{m}=0\}$. If $f_{1},f_{2}\in V$
and $f_{1}-f_{2}\in\cN_{m}$, we have $\Psi(f_{1})=\Psi(f_{2})+\Psi(f_{1}-f_{2})=\Psi(f_{2})$.
There can be two opposite interpretations of this fact:
\begin{enumerate}
\item The space $V=\Gamma_{\cpt}^{\infty}(D_{\re}M)$ is too large in the
sense that it contains a great deal of physically/empirically redundant
information, and instead we should use the quotient space $V/\cN_{m}$
as the ``test function space'' for the Majorana field;
\item The property $f_{1}-f_{2}\in\cN_{m}\then\Psi(f_{1})=\Psi(f_{2})$
is a nontrivial empirical law, which may be verified or falsified
experimentally. Hence this property does not imply that $V$ has some
empirically redundant information. However, this statement is not
precise since $\Psi(f_{i})$'s are not observables. Instead we can
say that $f_{1}-f_{2}\in\cN_{m}\then X_{f_{1},f}=X_{f_{2},f}$ is
an empirical law, for each $f\in V$.
\end{enumerate}
The validity of each interpretation depends on our knowledge on the
system. For example, consider the case where we know that the physical
system is a massive Majorana field, but do not know the value of the
mass $m>0$. Then for any given $m>0$, we do not know whether the
law $f_{1}-f_{2}\in\cN_{m}\then X_{f_{1},f}=X_{f_{2},f}$ holds; It
should be verified or falsified experimentally for each $m$. In this
case, the interpretation (2) is more plausible, and seems more realistic.
So we will adopt (2) here.

Thus now we have an example of an empirical law $f_{1}-f_{2}\in\cN_{m}\then X_{f_{1},f}=X_{f_{2},f}$
for the Majorana field, but this is rather qualitative. The simplest
but more quantitative example was already presented by Eq.~(\ref{eq:def:P(u'v'|uv)-infty}):
\begin{equation}
\Prob_{V}(f_{1},f_{1}'|f_{2},f_{2}')=\frac{1}{2}\left(1+(f_{1}|f_{2})_{m}(f_{1}'|f_{2}')_{m}-(f_{1}|f_{2}')_{m}(f_{1}'|f_{2})_{m}\right),\qquad(f_{1},f_{1}'),(f_{2},f_{2}')\in\ON_{2}(V).\label{eq:P(ff|ff)-law-V}
\end{equation}
Probably, the latter law (\ref{eq:P(ff|ff)-law-V}) ``almost implies''
the former law: For example, set $f_{1}'=f_{2}'=f$ and assume $f_{1}-f_{2}\in\cN_{m}$,
then $(f_{1}|f_{2})_{m}=1$%
, and hence
\[
f_{1}-f_{2}\in\cN_{m}\then\Prob_{V}\left(f_{2},f|f_{1},f\right)=1,\qquad\forall(f_{1},f),(f_{2},f)\in\ON_{2}(V).
\]
which means that the yes-no output of the measurement $P_{f_{1},f}^{+}$
is equal to that of $P_{f_{2},f}^{+}$, almost surely (i.e., in probability
1), in the sequential measurement $P_{f_{1},f}^{+}P_{f_{2},f}^{+}$.

{}

For the putative ``fundamental empirical law'' (\ref{eq:rho2-hol}),
concerning the 2-path Sorkin density function $\rho_{2}(\curve_{1},\curve_{2})$
for two piecewise smooth paths $\curve_{1},\curve_{2}$ such that
$\curve_{1}\bullet\curve_{2}^{-1}$ is closed, we consider an arbitrary
(sufficiently large) finite-dimensional subspace $W$ of $V/\cN_{m}$
with $\dim W=2n$, and%
{} the K{\"a}hler manifold $\tilde{\Manifold}:=\OCpl(W)\cong O(2n)/U(n)$
(or one of its connected components, $\Manifold:=\OCpl_{1}(W)\cong SO(2n)/U(n)$)
which is viewed as a finite-dimensional approximation of the infinite-dimensional
phase space for the classical Majorana field. Although the holonomy
$\hol(\curve)\in U(1)$ along a closed curve $\curve$ on $\Manifold$
can be defined even when $\Manifold$ is infinite-dimensional, the
rigorous justification of the path integral representation (\ref{eq:P(SS)=00003DpathInt})
is very difficult in those cases. Even if (\ref{eq:P(SS)=00003DpathInt})
is viewed as a \emph{formal} calculus of geometric quantization, ignoring
any measure-theoretical ground, the justification for general cases
(including the interacting fields) remains extremely difficult, because
it needs a rigorous formulation of renormalization. %
However, generally speaking, it is reasonable to believe that a finite-dimensional
phase spaces can give a good approximation for a quantum field theory;
Actually, the lattice field theory, which can be seen as a kind of
finite-dimensional approximation, appears to give a right approximation
for QFT.

Let us further explain the empirical meaning of this situation. Let
$\tilde{W}$ be a finite-dimensional subspace of $V$, such that $\dim\tilde{W}=2n$
and $\{f+\cN_{m}|f\in\tilde{W}\}=W$, so that $W\cong\tilde{W}$ (as
inner product spaces) by $f\mapsto f+\cN_{m}$. Clearly, there exists
a bounded region $O\subset M$ such that $\supp(f)\subset O$ for
all $f\in\tilde{W}$. Hence any self-adjoint element of $\CAR(\tilde{W})_{{\rm even}}$,
especially, $\Vac_{\ComplexS}\in\CAR(\tilde{W})_{{\rm even}}$ for
any $\ComplexS\in\OCpl(\tilde{W})$, is a local observable in $O$,
and hence we may consider that there is a measurement procedure for
it. Let $\OCpl_{i}(\tilde{W})$ ($i=1,2$) be the connected components
of $\OCpl(\tilde{W})$. If $\ComplexS_{i}\in\OCpl_{i}(\tilde{W})$
($i=1,2$), both $\Vac_{\ComplexS_{1}}$ and $\Vac_{\ComplexS_{2}}$
have empirical meanings. However, since $\Vac_{\ComplexS_{1}}\Vac_{\ComplexS_{2}}=0$
holds, and the set of projections $\{\Vac_{\ComplexS_{i}}|\ComplexS_{i}\in\OCpl_{i}(\tilde{W})\}$
have the same statistical law for $i=1,2$, we may consider only one
of $\OCpl_{i}(\tilde{W})$'s to calculate %
probabilities, without loss of generality. Thus we consider the connected
manifold $\Manifold:=\OCpl_{1}(\tilde{W})$. For a piecewise smooth
curve $\curve$ on $\Manifold$, the ``continuous measurement operation''
$A(\curve)$ is defined by (\ref{eq:def:A(C)}), as a limit of the
products of the projections $\Vac_{\ComplexS}$ ($\ComplexS\in\Manifold$).
We saw that this gives some empirical meaning to the holonomy $\hol(\curve_{1}\bullet\curve_{2}^{-1})\in U(1)$.%
{} For $S_{i}\in\Borel(\Manifold),\ i=1,...,N$, the probability $\Prob(S_{1},...,S_{N})$
was expressed by (\ref{eq:pathInt-hol}), which is a somewhat formal
expression, but can be justified rigorously.

Although $\Vac_{\ComplexS}$'s are considered to be the \emph{fundamental}
local observables, they seem far from ``familiar'' to us; It seems
even more difficult to conceive the realistic measurement procedure
of $\Vac_{\ComplexS}$, than that of $X_{f_{1},f_{2}}$. Probably,
this trade-off relation between the fundamentality and the familiarity
is unavoidable. For the theory of the concrete/realistic measurements
for the Majorana field, we will need further investigations on the
relation between this fundamental observables and more ``familiar''
ones, including the stress-energy tensor.

\providecommand{\noopsort}[1]{}\providecommand{\singleletter}[1]{#1}%

\end{document}